\documentclass[a4paper,11pt]{amsart}
\usepackage{amssymb}
\usepackage[cyr]{aeguill}
\usepackage{cmap}
\usepackage{hyperxmp}
\usepackage[pdfdisplaydoctitle=true,
            colorlinks=true,
            urlcolor=blue,
            citecolor=blue,
            linkcolor=blue,
            pdfstartview=FitH,
            pdfpagemode= UseNone,
            bookmarksnumbered=true]{hyperref}
\usepackage{xcolor}

\theoremstyle{plain}
\newtheorem{thm}{Theorem}[section]
\newtheorem{cor}[thm]{Corollary}
\newtheorem{lem}[thm]{Lemma}
\newtheorem{prop}[thm]{Proposition}
\newtheorem*{pb}{Problem}

\theoremstyle{definition}
\newtheorem{defn}[thm]{Definition}

\theoremstyle{remark}
\newtheorem{rem}[thm]{Remark}

\newtheorem*{expl}{Example}

\newcommand{\CC}{\mathbb{C}}
\newcommand{\NN}{\mathbb{N}}
\newcommand{\RR}{\mathbb{R}}
\newcommand{\ZZ}{\mathbb{Z}}
\newcommand{\Circle}{\mathbb{S}^{1}}

\newcommand{\PSL}{\mathrm{PSL}}
\newcommand{\psl}{\mathfrak{sl}}

\newcommand{\eps}{\varepsilon}
\newcommand{\id}{\mathrm{id}}
\newcommand{\e}{\mathbf{e}}
\newcommand{\riem}{\mathfrak{exp}}
\newcommand{\g}{\mathfrak{g}}
\newcommand{\gstar}{\mathfrak{g}^{*}}

\newcommand{\DiffS}{\mathrm{Diff}^{\infty}(\mathbb{S}^{1})}
\newcommand{\DiffSfix}[1]{\mathrm{Diff}_{#1}^{\infty}(\mathbb{S}^{1})}
\newcommand{\RotS}{\mathrm{Rot}(\mathbb{S}^{1})}

\newcommand{\D}[1]{\mathcal{D}^{#1}(\mathbb{S}^{1})}
\newcommand{\Dfix}[2]{\mathcal{D}_{#1}^{#2}(\mathbb{S}^{1})}

\newcommand{\VectS}{\mathrm{Vect}(\mathbb{S}^{1})}
\newcommand{\CS}{\mathrm{C}^{\infty}(\mathbb{S}^{1})}
\newcommand{\CSC}{\mathrm{C}^{\infty}(\mathbb{S}^{1},\CC)}
\newcommand{\CSfix}[1]{\mathrm{C}_{#1}^{\infty}(\mathbb{S}^{1})}

\newcommand{\HH}[1]{H^{#1}(\mathbb{S}^{1})}
\newcommand{\HHfix}[2]{H_{#1}^{#2}(\mathbb{S}^{1})}
\newcommand{\HHhat}[2]{\hat{H}_{#1}^{#2}(\mathbb{S}^{1})}

\newcommand{\norm}[1]{\left\Vert#1\right\Vert}
\newcommand{\abs}[1]{\left\vert#1\right\vert}
\newcommand{\set}[1]{\left\{#1\right\}}
\newcommand{\op}[1]{\mathbf{op}\left(#1\right)}
\newcommand{\sgn}{\mathop{\mathrm{sgn}}}

\DeclareMathOperator{\ad}{ad} %
\DeclareMathOperator{\Ad}{Ad} %
\DeclareMathOperator{\range}{im} %

\hypersetup{
    pdftitle={Right-invariant Sobolev metrics of fractional order on the diffeomorphism group of the circle}, 
    pdfauthor={J. Escher and B. Kolev}, 
    pdfsubject={MSC 2010: 58D05, 35Q53}, 
    pdfkeywords={Euler equation, diffeomorphism group of the circle, Sobolev metrics of fractional order}, 
    pdflang=EN 
    }

\begin{document}

\title[Right-invariant Sobolev metrics $H^{s}$]{Right-invariant Sobolev metrics of fractional order on the diffeomorphism group of the circle}

\author{Joachim Escher}
\address{Institute for Applied Mathematics, University of Hanover, D-30167 Hanover, Germany}
\email{escher@ifam.uni-hannover.de}

\author{Boris Kolev}
\address{Aix Marseille Universit\'{e}, CNRS, Centrale Marseille, I2M, UMR 7373, 13453 Marseille, France}
\email{boris.kolev@math.cnrs.fr}

\subjclass[2010]{58D05, 35Q53}
\keywords{Euler equation, diffeomorphism group of the circle, Sobolev metrics of fractional order} %

\date{August 26, 2014}

\begin{abstract}
In this paper, we study the geodesic flow of a right-invariant metric induced by a general Fourier multiplier on the diffeomorphism group of the circle and on some of its homogeneous spaces. This study covers in particular right-invariant metrics induced by Sobolev norms of fractional order. We show that, under a certain condition on the symbol of the inertia operator (which is satisfied for the fractional Sobolev norm $H^{s}$ for $s \ge 1/2$), the corresponding initial value problem is well-posed in the smooth category and that the Riemannian exponential map is a smooth local diffeomorphism. Paradigmatic examples of our general setting cover, besides all traditional Euler equations induced by a local inertia operator, the Constantin-Lax-Majda equation, and the Euler-Weil-Petersson equation.
\end{abstract}

\maketitle

\tableofcontents
\clearpage

\section{Introduction}
\label{sec:introduction}

The interest for geodesic flows on diffeomorphism groups goes back to Arnold~\cite{Arn1966}. He recast the Euler equations of hydrodynamics of an ideal fluid as the geodesic flow for the $L^{2}$ right invariant Riemannian metric on the volume preserving diffeomorphism group. Arnold's paper was, somehow, rather formal from the analytical point of view. The well-posedness of the geodesic flow was established, subsequently, by Ebin and Marsden in \cite{EM1970}. To do so, they introduced Hilbert manifolds of diffeomorphisms of class $H^{q}$, and used them to approximate the \emph{Fr\'{e}chet manifold} of smooth diffeomorphisms. This framework was extended thereafter to other equations of physical relevance \cite{CK2003,HMR1998,KLM2008,Kol2004,Mis1998,Shk1998,EKW2012}. Among these studies, right invariant metrics induced by $H^{k}$ Sobolev norms ($k \in \NN$) on the diffeomorphism group of the circle, $\DiffS$, have been extensively investigated \cite{CK2003,KLM2008,Kol2004,Shk1998}. In \cite{EKW2012}, well-posedness of the geodesic flow for the homogeneous $H^{1/2}$ right-invariant metric on the homogeneous space $\DiffS/\RotS$ was established. The homogeneous $H^{3/2}$ right-invariant metric on $\DiffS/\PSL(2,\RR)$ was considered in \cite{Gay2009}.

It is the aim of the present paper to study well-posedness of geodesic flows and their corresponding Euler equation for $H^{s}$ Sobolev norms on $\DiffS$, where $s \in \RR^{+}$. One of the main difficulty which arise immediately, in this context, is that the inertia operator is \emph{non-local}. More precisely, such a \emph{right-invariant} metric on $\DiffS$ is induced by an inner product
\begin{equation*}
    \langle u,v \rangle = \int_{\Circle} (Au)v \, dx,
\end{equation*}
on $\VectS = \CS$, where $A:\CS \to \CS$ is a (non-local) \emph{Fourier multiplier}. To be able to make use of the framework proposed by Ebin and Marsden, one needs to extend the metric \emph{smoothly} on $\D{q}$, the Hilbert manifold of diffeomorphisms of Sobolev class $H^{q}$. When $A$ is of finite order $r \ge 0$, it extends to a bounded linear operator from $\HH{q}$ to $\HH{q-r}$ for $q$ large enough, and the smoothness of the metric is reduced to the following question, where
\begin{equation*}
  R_{\varphi}: v \mapsto v \circ \varphi, \quad \varphi \in \D{q}, \, v \in \HH{q}.
\end{equation*}

\begin{pb}
Given a Fourier multiplier $A$ of order $r \ge 0$, under which conditions is the mapping
\begin{equation}\label{es:flat-map}
	\varphi \mapsto A_{\varphi}:= R_{\varphi} \circ A \circ R_{\varphi^{-1}} , \qquad \D{q} \to \mathcal{L}(\HH{q},\HH{q-r})
\end{equation}
smooth?
\end{pb}

Note that the problem is not trivial in general, because the mapping
\begin{equation*}
  (\varphi,v) \mapsto R_{\varphi} (v), \quad \D{q}\times\HH{q} \to \HH{q}
\end{equation*}
is \emph{not differentiable} (see \cite{EM1970} for instance), and the mapping $\varphi \mapsto R_{\varphi}$ is not even continuous (see remark~\ref{rem:norm-continuity}). It is however a well-known fact that when $A$ is a differential operator of order $r$, $A_{\varphi}$ is a linear differential operator whose coefficients are polynomial expressions of $1/\varphi_{x}$ and the derivatives of $\varphi$ up to order $r$ (e.g. $D_{\varphi} = (1/\varphi_x)D$), see~\cite{EM1970,EK2011} for instance. In that case, $\varphi \mapsto A_{\varphi}$ is smooth (in fact real analytic) for $q \ge r$.

However, for a general Fourier multiplier, we are not aware of any results in this direction. In theorem~\ref{thm:main_theorem}, however, we give a \emph{sufficient condition on the symbol} of $A$ which ensures that the mapping~\eqref{es:flat-map} is smooth. This answers a question raised in \cite[Appendix A]{EM1970}, at least in the case of the diffeomorphism group of the circle. Up to the authors knowledge, these results are new.

\begin{rem}
Of course, there are Fourier multiplication operators $A$, of order less than $1$, for which the mapping
\begin{equation*}
  \varphi \mapsto A_{\varphi}, \quad \D{q} \to \mathcal{L}(\HH{q},\HH{q-r})
\end{equation*}
is smooth. However, the present proof of theorem~\ref{thm:main_theorem} works only for $r \ge 1$. So far, the authors have not been able to exhibit a counter-example which would show that the conclusion of theorem~\ref{thm:main_theorem} is false for $0 \le r < 1$. They are not aware either of an example of a Fourier multiplier for which the conclusion of theorem~\ref{thm:main_theorem} fails for all $q \ge 0$.
\end{rem}

Theorem~\ref{thm:main_theorem} applies, in particular, to the inertia operator $\Lambda^{2s}$ of the Sobolev metric $H^{s}$ on $\DiffS$ for $s\in \RR$ and $s \ge 1/2$ (corollary~\ref{cor:Lambda}). This allows us to prove that the corresponding \emph{weak Riemannian} metric and its geodesic spray can be smoothly extended to the Hilbert manifold approximation $\D{q}$ for sufficiently large $q\in \RR$. As a corollary, we are able to prove local existence and uniqueness of geodesics on $\D{q}$ and $\DiffS$ (theorem~\ref{thm:smooth_flow}), as well as the well-posedness of the corresponding \emph{Euler equation} (corollary~\ref{cor:Euler-well-posedness}).

It is a well known result that the \emph{group exponential map} on $\DiffS$ is not locally surjective \cite{Mil1984}. In~\cite{CK2003}, it was shown, moreover, that the \emph{Riemannian exponential map} for the $L^{2}$ metric on $\DiffS$ was not a local diffeomorphism. However, due to the fact that the spray of the $H^{s}$ metric is smooth for $s \ge 1/2$ (theorem~\ref{thm:smoothness_spray}), we are able to prove that the exponential map on $\DiffS$ is a local diffeomorphism, in that case (theorem~\ref{thm:exponential_map}). From this fact, we can deduce that given two nearby diffeomorphisms, there is a unique geodesic which joins them and that this geodesic is a \emph{local minimum} of both the arc-length and the energy functionals. For $s=1/2$, this local minimum is, however, not a global minimum \cite{BBHM2013}. This exhibits a surprising difference with finite dimensional Riemannian geometry.

We close our study by extending our results to Euler equations on some homogeneous spaces of $\DiffS$, namely $\DiffS/\RotS$, where $\RotS$ is the subgroup of all rigid rotations of the circle $\Circle$ and $\DiffS/\PSL(2,\RR)$, where $\PSL(2,\RR)$ is the subgroup of all rigid M\"{o}bius transformations which preserve the circle $\Circle$. The first case includes the Constantin-Lax-Majda equation \cite{EKW2012,EW2012}. The second case includes the Euler-Weil-Petersson equation, which is related to the Weil-Petersson metric on the universal Teichm\"{u}ller space $T(1)$  \cite{NS1995, TT2006}.

The plan of the paper is as follows. In Section~\ref{sec:geometric_framework}, we recall basic materials on right-invariant metrics on the diffeomorphism group. Section~\ref{sec:smoothness-metric-and-spray} is devoted to the study of the smoothness of the extended metric an its spray on the Hilbert manifolds $\D{q}$. In Section~\ref{sec:existence_results}, we prove local existence and uniqueness of the initial value problem for the geodesics of the right-invariant $H^{s}$ metric on $\DiffS$ and well-posedness of the corresponding Euler equation. In Section~\ref{sec:exponential_map} we deal with the Riemannian exponential map and discuss the problem of minimization of the arc-length and the energy. In Section~\ref{sec:homogeneous_spaces} we extend our study to the homogeneous spaces $\DiffS/\RotS$ and $\DiffS/\PSL(2,\RR)$. We prove well-posedness for the corresponding Euler equations. In Appendix~\ref{sec:fourier_multipliers}, we prove some lemmas on Fourier multipliers, while in Appendix~\ref{sec:boundedness_properties}, we provide estimates and local boundedness properties for the right translation operator $R_{\varphi}$.

\section{Geometric Framework}
\label{sec:geometric_framework}

Let $\DiffS$ be the group of all smooth and orientation preserving diffeomorphisms on the circle. This group is equipped with a \emph{Fr\'{e}chet manifold} structure, modelled on the \emph{Fr\'{e}chet vector space} $\CS$ (see Guieu and Roger~\cite{GR2007}). Since, moreover, composition and inversion are smooth for this structure, we say that $\DiffS$ is a \emph{Fr\'{e}chet-Lie group}, cf. \cite{Ham1982}. Its Lie algebra, $\VectS$, is the space of smooth vector fields on the circle. It is isomorphic to $\CS$ with the Lie bracket given by
\begin{equation*}
    [u,w] = u_{x}w - uw_{x}.
\end{equation*}

Let $A:\CS \to \CS$ be a $L^{2}$-symmetric, positive definite, continuous linear operator on $\CS$, we define the following inner product on the Lie algebra $\VectS = \CS$
\begin{equation*}
    \langle u,w \rangle := \int_{\Circle} (Au)w \, dx = \int_{\Circle} u(Aw) \, dx.
\end{equation*}
Translating this inner product on each tangent space, we get one on each tangent space $T_{\varphi} \DiffS$, given by
\begin{equation}\label{eq:right-invariant-metric}
    \langle v_{1}, v_{2} \rangle_{\varphi} = \langle v_{1}\circ\varphi^{-1}, v_{2}\circ\varphi^{-1} \rangle_{\id} = \int_{\Circle} v_{1} (A_{\varphi}v_{2}) \varphi_{x} \,dx,
\end{equation}
where $v_{1}, v_{2} \in T_{\varphi} \DiffS$, $A_{\varphi} = R_{\varphi} \circ A \circ R_{\varphi^{-1}} $, and $R_{\varphi}(v):=v\circ\varphi$. One generates this way a \emph{smooth}, \emph{weak Riemannian metric} on $\DiffS$. For historical reasons going back to Euler~\cite{Eul1765a}, $A$ is called the \emph{inertia operator} of the right-invariant metric.

A \emph{covariant derivative} on a Fr\'{e}chet manifold is a way of differentiating vector fields along paths. In general, a torsion-free, covariant derivative, compatible with a \emph{weak Riemannian metric} does not exists but if it does, it is unique. According to Arnold, a necessary and sufficient for the existence of a covariant derivative compatible with a right-invariant metric on $\DiffS$ is the existence of the \emph{Arnold bilinear operator}
\begin{equation*}
    B(u,v) = \frac{1}{2}\Big( \ad_{u}^{\top} v + \ad_{v}^{\top} u \Big),
\end{equation*}
where $u,v \in \CS$ and $\ad_{u}^{\top}$ is the adjoint of the operator $\ad_{u}$, with respect to $A$ (see~\cite{Arn1966} for instance). If $\varphi(t)$ is any path in $\DiffS$ and $\xi(t)$ is a field of tangent vectors along the path, we define the covariant derivative along the path to be
\begin{equation*}
    D_{t} \xi = R_{\varphi} \left( w_{t} + \frac{1}{2}[u,w] + B(u,w) \right),
\end{equation*}
where $u(t) := \varphi_{t} \circ \varphi^{-1}$ and $w(t) := \xi(t) \circ \varphi^{-1}$. One can check that this covariant derivative is torsion-free and metric-compatible.

\begin{lem}
If $A:\CS \to \CS$ is invertible and commutes with $d/dx$, then, the map $\ad_{u}^{\top}$ is well defined and given by
\begin{equation*}
    \ad_{u}^{\top} w = A^{-1}\left[ 2(Aw)u_{x} + (Aw)_{x}u \right],
\end{equation*}
for $u,w \in \CS$.
\end{lem}

\begin{proof}
We have
\begin{equation*}
  <\ad_{u}v,w> = \int_{\Circle} (Aw)(u_{x}v-uv_{x})\,dx = \int_{\Circle} \left[2(Aw)u_{x} + (Aw)_{x}u)\right]v\,dx
\end{equation*}
where $u,v,w \in\CS$. But since $A :\CS \to \CS$ is invertible, we get, finally
\begin{equation*}
    \ad_{u}^{\top} w = A^{-1}\left[ 2(Aw)u_{x} + (Aw)_{x}u \right].
\end{equation*}
\end{proof}

A \emph{geodesic} is a path $\varphi(t)$ in $\DiffS$, which is an extremal curve of the energy functional
\begin{equation*}
  \mathcal{E} := \frac{1}{2}\int_{0}^{1}  \langle u(t),u(t) \rangle\, dt ,
\end{equation*}
where $u(t) = \varphi_{t} \circ \varphi^{-1}$. The corresponding Euler-Lagrange equation
\begin{equation*}
  D_{t}\varphi_{t}=0
\end{equation*}
is equivalent to the following first order equation
\begin{equation}\label{eq:Euler_equation}
    u_{t} = -B(u,u) = -A^{-1}\left\{(Au)_{x}u+2(Au)u_{x}\right\},
\end{equation}
called the \emph{Euler equation}.

\begin{expl}
For the $L^{2}$-metric ($A=I$), the corresponding Euler equation~\eqref{eq:Euler_equation} is the \emph{inviscid Burgers equation}
\begin{equation*}
    u_{t} + 3uu_{x} = 0.
\end{equation*}
\end{expl}

\begin{expl}
For the $H^{1}$-metric ($A= I-d^{2}/dx^{2}$), the corresponding Euler equation~\eqref{eq:Euler_equation} is the \emph{Camassa-Holm equation}
\begin{equation*}
    u_{t} - u_{txx} + 3uu_{x} - 2u_{x}u_{xx} - uu_{xxx} = 0,
\end{equation*}
a model in the theory of shallow water waves \cite{CH1993,FF1981/82}.
\end{expl}

Let $\varphi$ be the flow of the time dependent vector field $u$ \textit{i.e.}, $\varphi_{t} = u \circ \varphi$ and let $v = \varphi_{t}$. Then $u$ solves the Euler equation~\eqref{eq:Euler_equation}, if and only if, $(\varphi,v)$ is a solution of
\begin{equation}\label{eq:ODE_spray}
\left\{\begin{aligned}
    \varphi_{t} &= v, \\
    v_{t} &= S_{\varphi}(v),
\end{aligned}
\right.
\end{equation}
where
\begin{equation*}
  S_{\varphi}(v):=\left(R_{\varphi}\circ S\circ R_{\varphi^{-1}} \right)(v),
\end{equation*}
and
\begin{equation*}
  S(u):= A^{-1}\left\{ [A,u]u_{x}-2(Au)u_{x}\right\}.
\end{equation*}
The \emph{second order vector field} on $\DiffS$, defined by
\begin{equation}\label{eq:geodesic-spray}
    F : (\varphi,v)\mapsto \left(\varphi,v,v, S_{\varphi}(v)\right)
\end{equation}
is called the \emph{geodesic spray}, following Lang~\cite{Lan1999}.

Suppose now that $A$ is a differential operator of order $r$. Then, the right hand side of the Euler equation is of \emph{order 1}. It is however quite surprising that, in Lagrangian coordinates, the propagator of evolution equation of the geodesic flow has better mapping properties, provided that the \emph{order} $r$ of $A$ is not less than 1. Indeed, in that case, the quadratic operator
\begin{equation*}
    S(u):= A^{-1}\left\{ [A,u]u_{x}-2(Au)u_{x}\right\}
\end{equation*}
is of order $0$ because the commutator $[A,u]$ is of order less than $\le r-1$. One might expect, that for a larger class of operators $A$, the quadratic operator $S$ to be of order $0$ and the second order system \eqref{eq:ODE_spray} to be the local expression of an ODE on some suitable Banach manifold.

This observation is at the root of a strategy proposed in the 70' by Ebin and Marsden~\cite{EM1970} to study well-posedness of the Euler equation. Following their approach, if we can prove \emph{local existence and uniqueness of geodesics} (ODE) on diffeomorphism groups then the PDE (Euler equation) is \emph{well-posed}. To do so, it is necessary to introduce an approximation of the Fr\'{e}chet--Lie group $\DiffS$ by \emph{Hilbert manifolds}. Let $\HH{q}$ be the completion of $\CS$ for the norm
\begin{equation*}
  \norm{u}_{H^{q}}:= \left( \sum_{k \in \ZZ}(1 + k^{2})^{q}\abs{\hat{u}_{k}}^{2} \right)^{1/2},
\end{equation*}
where $q\in \RR, q \ge 0$. We recall that $\HH{q}$ is a multiplicative algebra for $q > 1/2$ (cf. \cite[Theorem 2.8.3]{Tri1983}). This means that
\begin{equation*}
  \norm{uv}_{\HH{q}} \lesssim \norm{u}_{\HH{q}} \norm{v}_{\HH{q}}, \quad u,v \in \HH{q}.
\end{equation*}

\begin{defn}
We say that a $C^{1}$ diffeomorphism $\varphi$ of $\Circle$ is of class $H^{q}$ if for any of its lifts to $\RR$, $\tilde{\varphi}$, we have
\begin{equation*}
  \tilde{\varphi} - \mathrm{id} \in \HH{q}.
\end{equation*}
\end{defn}

For $q > 3/2$, the set $\D{q}$ of $C^{1}$-diffeomorphisms of the circle which are of class $H^{q}$ has the structure of a \emph{Hilbert manifold}, modelled on $\HH{q}$ (see~\cite{EM1970}). The manifold $\D{q}$ is also a \emph{topological group} but \emph{not a Lie group} (composition and inversion in $\D{q}$ are continuous but \emph{not differentiable}). Note however, that, given $\varphi \in \D{q}$,
\begin{equation*}
	u \mapsto R_{\varphi}(u) := u \circ \varphi, \qquad \HH{q} \to \HH{q}
\end{equation*}
is a \emph{smooth map}, and that
\begin{equation*}
	(u,\varphi) \mapsto u \circ \varphi, \qquad \HH{q+k} \times \D{q} \to \HH{q}
\end{equation*}
is of class $C^{k}$.

The Fr\'{e}chet Lie group group $\DiffS$ may be viewed as an inverse limit of \emph{Hilbert manifolds} (ILH)
\begin{equation*}
    \DiffS = \bigcap_{q > \frac{3}{2}} \D{q},
\end{equation*}
and we call the scales of manifolds $\D{q})_{q > 3/2}$, a Hilbert manifold approximation of $\DiffS$.

\begin{rem}
Note that the tangent bundle of the Hilbert manifold $\D{q}$ is trivial. Indeed, let $\mathfrak{t}: T\Circle \to \Circle \times \RR$ be a trivialisation of the tangent bundle of the circle. Then
\begin{equation*}
  \Psi : T\D{q} \to \D{q}\times \HH{q}, \qquad \xi \mapsto \mathfrak{t} \circ \xi
\end{equation*}
defines a smooth vector bundle isomorphism because $\mathfrak{t}$ is smooth (see~\cite[Page~107]{EM1970}).
\end{rem}

Within this framework, the case where the inertia operator $A$ is a \emph{differential} operator with constant coefficients has been extensively studied in the literature (see for instance \cite{CK2002,CK2003,EK2011}). It is the aim of the present paper to extend these results when $A$ is a general \emph{Fourier multiplier}, that is, a continuous linear operator on $\CS$, which commutes with $D:=d/dx$. In that case, we get
\begin{equation*}
    (Au)(x) = \sum_{k\in\mathbb{Z}} a(k)\hat{u}(k) \exp (2i\pi k x),
\end{equation*}
where $\hat{u}(k)$ is the $k$-th Fourier coefficients of $u$ (see lemma~\ref{lem:FOp}). The sequence $a: \mathbb{Z}\to\mathbb{C}$ is called the \emph{symbol} of $A$ and we shall use the notation $A = \op{a(k)}$. When $a(k) = \mathcal{O}(\abs{k}^{r})$, the Fourier multiplier $A=\op{a(k)}$ extends, for each $q \ge r$, to a bounded linear operator in $\mathcal{L}(\HH{q},\HH{q-r})$. It is said to be of order $r$.

\begin{expl}
The inertia operator for the $H^{s}$ Sobolev metric ($s \ge 1/2$), defined by
\begin{equation*}
  \Lambda_{2s}:= \op{(1+k^{2}}^{s})
\end{equation*}
is of order $2s \ge 1$.
\end{expl}

\section{Smoothness of the extended metric and spray}
\label{sec:smoothness-metric-and-spray}

The fact that a right-invariant metric, defined by~\eqref{eq:right-invariant-metric}, and its geodesic spray $F$, defined by~\eqref{eq:geodesic-spray} are smooth on $\DiffS$ is unfortunately useless to establish the well-posedness of the geodesic flow. What we need to do is to study under which conditions, the metric and its spray \emph{can be extended smoothly} to the Hilbert approximation manifolds $\D{q}$. In this section, we provide a criteria on the inertia operator $A$ (satisfied by almost all known examples) which ensures the smoothness of the metric on the extended manifolds $\D{q}$, for $q$ large enough.

For general materials on Banach manifolds, we refer to \cite{Lan1999}. Let $X$ be a Banach manifold modelled over a Banach space $E$. We recall that a Riemannian metric $g$ on $X$ is a smooth, symmetric, positive definite, covariant $2$-tensor field on $X$. In other words, we have for each $x \in X$ a symmetric, positive definite, bounded, bilinear form $g(x)$ on $T_{x}X$ and, in any local chart $U$, the mapping
\begin{equation*}
  x \to g(x), \qquad U \to \mathcal{L}_{\mathrm{sym}}^{2}(E,\RR)
\end{equation*}
is smooth. Given any $x\in X$, we can then consider the bounded, linear operator
\begin{equation*}
  h_{x} : T_{x}X \to T_{x}^{*}X,
\end{equation*}
called the \emph{flat map} and defined by $h_{x}(\xi_{x}) = g(x)(\xi_{x}, \cdot)$. The metric is \emph{strong} if $h_{x}$ is a topological linear isomorphism for all $x \in X$, whereas it is \emph{weak} if $h_{x}$ is only injective for all $x \in X$.

Given $A \in \mathrm{Isom}(\HH{q},\HH{q-r})$, it induces a \emph{bounded inner product} on each tangent space, $T_{\varphi}\D{q}$ for each $\varphi \in \D{q}$, given by
\begin{equation*}
    \langle v_{1}, v_{2} \rangle_{\varphi} = \int_{\Circle} v_{1} (A_{\varphi}v_{2}) \varphi_{x} \,dx,
\end{equation*}
where $A_{\varphi}:= R_{\varphi} \circ A \circ R_{\varphi^{-1}}$. To conclude, however, that the family $\langle \cdot , \cdot \rangle_{\varphi}$ defines a (weak) Riemannian metric on the Banach manifold $\D{q}$, we need to show that the mapping
\begin{equation*}
  \varphi \mapsto \varphi_{x}A_{\varphi}, \quad \D{q} \to \mathcal{L}^{2}(\HH{q},\HH{-q}),
\end{equation*}
is smooth.

\begin{rem}
Note that even when the flat map
\begin{equation*}
  \tilde{A} : (\varphi,v) \mapsto (\varphi,\varphi_{x}A_{\varphi}v), \quad \D{q}\times \HH{q} \to \D{q}\times \HH{-q}
\end{equation*}
is smooth and defines an injective vector bundle morphism, its image
\begin{equation*}
  \D{q}\times \HH{q-r}
\end{equation*}
is not a subbundle of $T^{*}\D{q}$ in the sense of \cite[III.3]{Lan1999}, because $\HH{q-r}$ is \emph{not a closed subspace} of $\HH{-q}$.
\end{rem}

For $q > 3/2$, the mappings $\varphi \mapsto \varphi_{x}$ and $\varphi \mapsto 1/\varphi_{x}$ are smooth from $\D{q} \to \HH{q-1}$. Thus, for $r \ge 1$ and $q-r \ge 0$, lemma~\ref{lem:pointwise_multiplication} shows that the metric is smooth if and only if
\begin{equation*}
	\varphi \mapsto A_{\varphi}, \qquad \D{q} \to \mathcal{L}(\HH{q},\HH{q-r})
\end{equation*}
is smooth. If this holds, we can compute, for each $n \ge 1$, the $n$-th Fr\'{e}chet differential\footnote{We have chosen to denote by $\partial$ the Fr\'{e}chet differential to avoid the confusion with the already used notation $D = d/dx$.}
\begin{equation*}
	\partial^{n}_{\varphi}A_{\varphi} \in \mathcal{L}^{n+1}(\HH{q},\HH{q-r}).
\end{equation*}
which is itself smooth.

\begin{lem}\label{lem:nth_derivative}
We have
\begin{equation}\label{eq:nth_derivative}
    \partial^{n}_{\varphi}A_{\varphi} (v,\delta\varphi_{1}, \dotsc ,\delta\varphi_{n}) = R_{\varphi}A_{n}R_{\varphi}^{-1}(v,\delta\varphi_{1}, \dotsc ,\delta\varphi_{n}),
\end{equation}
where
\begin{equation*}
  A_{n}:= \partial_{\id}A_{\varphi} \in \mathcal{L}^{n+1}(\HH{q},\HH{q-r})
\end{equation*}
is the $(n+1)$-linear operator defined inductively by $A_{0} = A$ and
\begin{multline}\label{eq:recurrence_relation}
    A_{n+1}(u_{0},u_{1}, \dotsc , u_{n+1}) = \nabla_{u_{n+1}} A_{n}(u_{0}, u_{1}, \dotsc , u_{n}) \\
    - \sum_{k=0}^{n}A_{n}(u_{0}, \dotsc ,\nabla_{u_{n+1}} u_{k}, \dotsc , u_{n}),
\end{multline}
where $\nabla$ is the canonical connection on the Lie group $\Circle$.
\end{lem}

\begin{rem}
For $n=1$, we get
\begin{equation*}
    A_{1}(u_{0},u_{1}) = [\nabla_{u_{1}},A]u_{0},
\end{equation*}
and for $n=2$, we get
\begin{equation*}
    A_{2}(u_{0},u_{1},u_{2}) = \big( [\nabla_{u_{2}},[\nabla_{u_{1}},A]] - [\nabla_{\nabla_{u_{2}}u_{1}},A] \big) u_{0}.
\end{equation*}
\end{rem}

\begin{proof}
We will make the computations for smooth functions. The general result follows from a density argument and the fact that the expressions are continuous. Formula~\eqref{eq:nth_derivative} is trivially true for $n=0$. Now suppose it is true for some $n\in\NN$, that is
\begin{equation*}
    \partial^{n}_{\varphi}A_{\varphi} (v,\delta\varphi_{1}, \dotsc ,\delta\varphi_{n}) = R_{\varphi}A_{n}R_{\varphi}^{-1}(v,\delta\varphi_{1}, \dotsc ,\delta\varphi_{n}),
\end{equation*}
Let $\varphi(s)$ be a smooth path in $\DiffS$ such that
\begin{equation*}
    \varphi(0) = \varphi, \qquad \partial_{s} \; \varphi(s) \big|_{s = 0} = \delta\varphi_{n+1}.
\end{equation*}
Set $u_{k} = \delta\varphi_{k}\circ\varphi^{-1}$, for $1 \le k \le n+1$ and $u_{0} = v\circ\varphi^{-1}$. We compute first
\begin{equation*}
	\partial_{s} \left\{ R_{\varphi(s)}w \right\}_{s = 0} = R_{\varphi} \left(u_{n+1}w_{x}\right),
\end{equation*}
for $w\in\CS$, and
\begin{equation*}
    \partial_{s} \left\{ R_{\varphi(s)}^{-1}w \right\}_{s = 0} = - u_{n+1} \left(R_{\varphi}^{-1}w\right)_{x},
\end{equation*}
for $w\in\CS$. Hence
\begin{multline*}
    \partial_{s} \left\{ R_{\varphi} A_{n}R_{\varphi}^{-1}(v,\delta\varphi_{1}, \dotsc ,\delta\varphi_{n}) \right\}_{s = 0} = \\
    R_{\varphi} \left\{ u_{n+1} (A_{n}(u_{0}, \dotsc ,u_{n}))_{x} \right\} - \sum_{k=0}^{n} R_{\varphi}A_{n}(u_{0},\dotsc , u_{n+1}(u_{k})_{x}, \dotsc ,u_{n}),
\end{multline*}
which gives the recurrence relation~\eqref{eq:recurrence_relation}, since
\begin{equation*}
    u_{n+1} (A_{n}(u_{0}, \dotsc ,u_{n}))_{x} = \nabla_{u_{n+1}} A_{n}(u_{0}, u_{1}, \dotsc , u_{n}),
\end{equation*}
and
\begin{equation*}
    A_{n}(u_{0}, \dotsc , u_{n+1}(u_{k})_{x}, \dotsc ,u_{n}) = A_{n}(u_{0}, \dotsc ,\nabla_{u_{n+1}} u_{k}, \dotsc , u_{n}).
\end{equation*}
\end{proof}

We shall prove now the following \emph{necessary and sufficient condition} for smoothness.

\begin{thm}[Smoothness Theorem]\label{thm:smoothness_theorem}
Let $A : \CS \to \CS$ be a continuous linear operator of order $r \ge 1$. Let $q > 3/2$ with $q-r \ge 0$. Then
\begin{equation*}
	\varphi \mapsto A_{\varphi}:= R_{\varphi} \circ A \circ R_{\varphi^{-1}} , \quad \D{q} \to \mathcal{L}(\HH{q},\HH{q-r})
\end{equation*}
is smooth, if and only if, each $A_n$ extends to a bounded $(n+1)$-linear operator in $\mathcal{L}^{n+1}(\HH{q},\HH{q-r})$.
\end{thm}

The idea of the proof of theorem~\ref{thm:smoothness_theorem}, which is inductive, is the following. First, we show that if $A_n$ is bounded, then the mapping
\begin{equation*}
  \varphi \mapsto A_{n,\varphi} := R_{\varphi}A_{n}R_{\varphi^{-1}}, \quad \D{q} \to \mathcal{L}^{n+1}(\HH{q},\HH{q-r})
\end{equation*}
is locally bounded. Then, we prove that if $A_{n+1,\varphi}$ is locally bounded, then $A_{n,\varphi}$ is locally Lipschitz. Finally we show that
if $A_{n+1,\varphi}$ is locally Lipschitz, then $A_{n,\varphi}$ is $C^{1}$. The full detail proof, given below, requires the following two elementary lemmas, which will be stated without proof.

\begin{lem}\label{lem:path_integral_continuity}
Let $X$ be a topological space and $E$ a Banach space. Let $f: [0,1] \times X\to E$ be a continuous mapping. Then the mapping
\begin{equation*}
  g(x):= \int_{0}^{1} f(t,x)\, dt
\end{equation*}
is continuous.
\end{lem}

\begin{lem}\label{lem:mean_value_criteria}
Let $E$, $F$ be Banach spaces and $U$ a convex, open set in $E$. Let $\alpha: U \to \mathcal{L}(E,F)$ be a continuous mapping and $f : U \to F$ a mapping such that
\begin{equation*}
  f(y) - f(x) = \int_{0}^{1} \alpha(ty + (1-t)x)(y-x)\, dt,
\end{equation*}
for all $x,y \in U$. Then $f$ is $C^{1}$ on $U$ and $df = \alpha$.
\end{lem}

\begin{proof}[Proof of theorem~\ref{thm:smoothness_theorem}]
Note first that for $q > 3/2$ and $q-r \ge 0$, the mapping
\begin{equation*}
  \varphi \mapsto R_{\varphi}, \quad \D{q} \to \mathcal{L}(\HH{q-r},\HH{q-r})
\end{equation*}
is locally bounded (lemma~\ref{lem:Rphi_estimate}) and that the mapping
\begin{equation*}
  (\varphi,v) \mapsto v \circ \varphi, \quad \D{q} \times \HH{q-r} \to \HH{q-r}
\end{equation*}
is continuous (corollary~\ref{cor:Rphi_continuity}). Since, moreover, the mapping
\begin{equation*}
  \varphi \mapsto \varphi^{-1}, \quad \D{q} \to \D{q}
\end{equation*}
is continuous for $q > 3/2$ (see \cite{EM1970,IKT2013}), we conclude that
\begin{equation*}
  \varphi \mapsto R_{\varphi^{-1}}, \quad \D{q} \to \mathcal{L}(\HH{q},\HH{q})
\end{equation*}
is locally bounded and that
\begin{equation*}
  (\varphi,v) \mapsto v \circ \varphi^{-1}, \quad \D{q} \times \HH{q} \to \HH{q}
\end{equation*}
is continuous.

The fact that the boundedness of the $A_{n}$ is a necessary condition results from lemma~\ref{lem:nth_derivative}. Conversely, suppose that each $A_{n}$ is bounded in $\mathcal{L}^{n+1}(\HH{q},\HH{q-r})$, then
\begin{equation*}
  A_{n,\varphi} := R_{\varphi}A_{n}R_{\varphi^{-1}} \in \mathcal{L}^{n+1}(\HH{q},\HH{q-r}),
\end{equation*}
is a bounded $(n+1)$-linear operator, with
\begin{multline*}
  \norm{A_{n,\varphi}}_{\mathcal{L}^{n+1}(H^{q},H^{q-r})} \le \\
  \norm{R_{\varphi}}_{\mathcal{L}(H^{q-r},H^{q-r})} \norm{A_{n}}_{\mathcal{L}^{n+1}(H^{q},H^{q-r})} \norm{R_{\varphi^{-1}}}_{\mathcal{L}(H^{q},H^{q})}^{n+1},
\end{multline*}
and we conclude, thus, that
\begin{equation*}
  \varphi \mapsto A_{n,\varphi}, \quad \D{q} \to \mathcal{L}^{n+1}(\HH{q},\HH{q-r})
\end{equation*}
is locally bounded, for each $n \in \NN$.

(a) We will first show that $\varphi \mapsto A_{n,\varphi}$ is \emph{locally Lipschitz continuous}\footnote{On $\D{q}$, we did not introduce any distance compatible with the topology. The concept of a \emph{locally Lipschitz mapping} $f : M \to E$, from a Banach manifold $M$ to a Banach vector space $E$ does not require such an additional structure. It is defined using a local chart, and then shown to be independent of the choice of the particular chart.}, for each $n \in \NN$. Let $\psi\in\D{q}$ be given. Because $\varphi \mapsto A_{n+1,\varphi}$ is \emph{locally bounded}, it is possible to find a neighbourhood $U$ of $\psi$ and a positive constant $K$ such that
\begin{equation*}
    \norm{A_{n+1,\varphi}}_{\mathcal{L}^{n+1}(H^{q},H^{q-r})} \le K, \quad \forall \varphi \in U.
\end{equation*}
We can further assume (using a local chart) that $U$ is a ball in $\HH{q}$. Pick now $\varphi_{0}$ and $\varphi_{1}$ in $\DiffS \cap U$ and set $\varphi(t):=(1-t)\varphi_{0}+t\varphi_{1}$ for $t\in [0,1]$. Choosing $v_{0},\dotsc,v_{n}\in\CS$ with $\norm{v_j}_{H^{q}}\le 1$, we obtain from lemma~\ref{lem:nth_derivative} that
\begin{equation*}
    A_{n, \varphi_{1}}(v_{0}, \dotsc ,v_{n}) - A_{n, \varphi_{0}}(v_{0}, \dotsc ,v_{n}) =
    \int_{0}^{1} A_{n+1, \varphi(t)}(v_{0}, \dotsc ,v_{n}, \varphi_{1}-\varphi_{0} )\, dt.
\end{equation*}
This implies
\begin{equation*}
   \norm{A_{n, \varphi_{1}}(v_{0}, \dotsc ,v_{n}) - A_{n, \varphi_{0}}(v_{0}, \dotsc ,v_{n})}_{H^{q-r}} \le
    K \norm{ \varphi_{1}-\varphi_{0}}_{H^{q}},
\end{equation*}
for all $v_{0},\dotsc,v_{n}\in\CS$ with $\norm{v_j}_{H^{q}}\le 1$. The assertion that $A_{n, \varphi}$ is Lipschitz continuous follows from the density of the embedding $\CS\hookrightarrow\HH{q}$, and continuity of the mapping
\begin{equation*}
  (\varphi,v_{0},\dotsc,v_{n}) \mapsto A_{n,\varphi}(v_{0},\dotsc,v_{n}), \quad \D{q}\times \HH{q}^{n+1} \to \HH{q-r}.
\end{equation*}

(b) We will now show by induction, that $\varphi \mapsto A_{\varphi}$ is of class $C^{n}$ for all $n \in \NN$, and that its $n$-th Fr\'{e}chet derivative is $A_{n,\varphi}$. For each $n \ge 1$, let
\begin{equation*}
  \alpha_{n}: \D{q} \to \mathcal{L}\Big(\HH{q},\mathcal{L}^{n}\big(\HH{q},\HH{q-r}\big)\Big),
\end{equation*}
be the \emph{Lipschitz continuous} mapping defined by
\begin{equation*}
  \alpha_{n}(\varphi):= \big[\delta\varphi_{n} \mapsto A_{n,\varphi}(\cdot,\dotsc,\cdot,\delta\varphi_{n})\big].
\end{equation*}
Let $U$ be a local chart in $\D{q}$, that we choose to be a convex open subset of $\HH{q}$. By its very definition, we have
\begin{equation*}
  A_{\varphi_{1}}(v) - A_{\varphi_{0}}(v) = \int_{0}^{1} A_{1, t\varphi_{1} + (1-t)\varphi_{0}}(v,\varphi_{1}-\varphi_{0})\, dt,
\end{equation*}
for all $\varphi_{0},\varphi_{1}\in U\cap\CS$ and $v \in \CS$. But, the continuity of the mapping $\varphi \mapsto A_{1,\varphi}$, together with lemma~\ref{lem:path_integral_continuity}, and the density of the embedding $\CS\hookrightarrow\HH{q}$, permit to conclude that this formula is still true for all $\varphi_{0}, \varphi_{1} \in U$ and $v\in \HH{q}$. Therefore, we can write in $\mathcal{L}(\HH{q},\HH{q-r})$
\begin{equation*}
  A_{\varphi_{1}} - A_{\varphi_{0}} = \int_{0}^{1} \alpha_{1}(t\varphi_{1} + (1-t)\varphi_{0})(\varphi_{1}-\varphi_{0})\, dt,
\end{equation*}
and, by virtue of lemma~\ref{lem:mean_value_criteria}, we conclude that $\varphi \mapsto A_{\varphi}$ is $C^{1}$ and that $DA_{\varphi} = \alpha_{1}$. A similar argument shows that, for each $n \ge 1$, we have
\begin{equation*}
  \alpha_{n}(\varphi_{1}) - \alpha_{n}(\varphi_{0}) = \int_{0}^{1} \alpha_{n+1}(t\varphi_{1} + (1-t)\varphi_{0})(\varphi_{1}-\varphi_{0})\, dt,
\end{equation*}
and hence that $\alpha_{n}$ is $C^{1}$ with $D\alpha_{n} =\alpha_{n+1}$. This completes the proof.
\end{proof}

When $A$ is a Fourier multiplier, a criteria on the symbol $a$ of $A$ which ensures that all the $A_{n}$ are bounded and thus that the metric is smooth is given below. It's proof is a direct consequence of lemma~\ref{lem:nth_derivative_symbol}, lemma~\ref{lem:n_order_estimate} and corollary~\ref{cor:nth_derivative_estimates} in Appendix~\ref{sec:fourier_multipliers}.

\begin{thm}\label{thm:main_theorem}
  Let $A=\op{a(k)}$ be a Fourier multiplier of order $r \ge 1$. Suppose that its symbol $a$ extends to $\RR$ and that for each $n\ge 1$, the function
  \begin{equation*}
    f_{n}(\xi) := \xi^{n-1}a(\xi)
  \end{equation*}
  is of class $C^{n-1}$, that $f_n^{(n-1)}$ is absolutely continuous and that there exists $C_{n}>0$ such that
  \begin{equation}\label{eq:nth_derivative_condition}
    \abs{f_n^{(n)}(\xi)} \le C_{n} (1+\xi^{2})^{(r-1)/2},
  \end{equation}
  almost everywhere. Then,
  \begin{equation*}
	\varphi \mapsto A_{\varphi}:= R_{\varphi} \circ A \circ R_{\varphi^{-1}} , \quad \D{q} \to \mathcal{L}(\HH{q},\HH{q-r})
  \end{equation*}
  is smooth for $q > 3/2$ and $q-r \ge 0$.
\end{thm}

\begin{rem}
This criteria is always satisfied when the symbol $A=\op{a(k)}$ belongs to the class $\mathcal{S}^{r}$, that is when $a$ can be extended to a smooth function on $\RR$ such that
\begin{equation}\label{eq:symbol-estimate}
  a^{(k)}(\xi) = \mathcal{O}(\abs{\xi}^{r-k}), \quad \forall k \in \NN.
\end{equation}
\end{rem}

This applies in particular to the inertia operator $\Lambda^{2s}$ of the Sobolev metric $H^{s}$, when $s \ge 1/2$. Indeed, let $a_{s}(\xi):= \left(1+ \xi^{2}\right)^{s}$ be the symbol of $\Lambda^{2s}$. One can check that
\begin{equation*}
    a_{s}^{(k)}(\xi) = \frac{p_{k}(\xi)}{(1+\xi^{2})^{k}}a_{s}(\xi),
\end{equation*}
for $k \ge 1$, where $p_{k}$ is a polynomial function with $d(p_{k}) \le k$. Thus, \eqref{eq:symbol-estimate} is true for $a_{s}$, and we have the following result.

\begin{cor}\label{cor:Lambda}
Let $s \in \RR$ and $\Lambda^{2s} := \op{ \left( 1+ n^{2} \right)^{s} }$. If $s \ge 1/2$ then the mapping
\begin{equation*}
	\varphi \mapsto R_{\varphi}\circ\Lambda^{2s}\circ R_{\varphi^{-1}}, \quad \D{q} \to \mathcal{L}(\HH{q},\HH{q-2s})
\end{equation*}
is smooth for $q > 3/2$ and $q-2s \ge 0$.
\end{cor}

On a finite dimensional manifold, as soon as the metric is $C^{k}$, the geodesic spray is $C^{k-1}$, because the components of the spray, in any local chart, involve the \emph{Christoffel symbols} which depend on the first derivatives of the metric. We might, therefore, expect some kind of analog results to hold for a weak Riemannian metric on a Banach manifold, as soon as the spray exists.

\begin{thm}\label{thm:smoothness_spray}
Let $A$ be a Fourier multiplier of order $r \ge 1$ and let $q > 3/2$, with $q-r \ge 0$. Suppose, moreover, that
\begin{equation*}
  \varphi \mapsto A_{\varphi} = R_{\varphi} \circ A \circ R_{\varphi^{-1}}, \quad \D{q} \to \mathrm{Isom}(\HH{q},\HH{q-r})
\end{equation*}
is smooth. Then the geodesic spray
\begin{equation*}
	(\varphi, v) \mapsto S_{\varphi} (v) = R_{\varphi} \circ S \circ R_{\varphi^{-1}} (v),
\end{equation*}
where
\begin{equation*}
    S(u) = A^{-1} \left\{ [A,u] u_{x} - 2(A u)u_{x} \right\}.
\end{equation*}
extends smoothly to $T\D{q} = \D{q}\times \HH{q}$.
\end{thm}

\begin{proof}
Let  $P(u) := (A u) u_{x}$ and $Q(u) := [A,u] u_{x}$. Then,
\begin{equation*}
    S_{\varphi}(v) = A_{\varphi}^{-1} \left\{ Q_{\varphi}(v) - 2P_{\varphi}(v) \right\}.
\end{equation*}
The proof reduces to establish, using the chain rule, that the three mappings
\begin{equation*}
  (\varphi,v) \mapsto P_{\varphi}(v), \quad (\varphi,v) \mapsto Q_{\varphi}(v), \quad (\varphi,w) \mapsto A_{\varphi}^{-1} (w)
\end{equation*}
are smooth.

(a) We have $P_{\varphi}(v) = \big( A_{\varphi} (v) \big)\big(  D_{\varphi}(v) \big)$. But
\begin{equation*}
	(\varphi,v) \mapsto  A_{\varphi}(v), \quad \D{q}\times\HH{q} \to \HH{q-r}
\end{equation*}
is smooth by hypothesis, whereas
\begin{equation*}
	(\varphi,v) \mapsto  D_{\varphi}(v),\quad \D{q}\times\HH{q} \to \HH{q-1}
\end{equation*}
is smooth since $D_{\varphi}(v) = v_{x}/\varphi_{x}$ and $\HH{q-1}$ is a multiplicative algebra for $q > 3/2$. To conclude that
\begin{equation*}
	(\varphi,v) \mapsto  P_{\varphi}(v), \quad \D{q}\times\HH{q} \to \HH{q-r}
\end{equation*}
is smooth, we use the fact that pointwise multiplication extends to a bounded bilinear mapping
\begin{equation*}
  \HH{q-r} \times \HH{q-1} \to \HH{q-r},
\end{equation*}
if $q-1 > 1/2$ and $0 \le q-r \le q-1$ (c.f. lemma~\ref{lem:pointwise_multiplication}).

(b) By virtue of lemma~\ref{lem:nth_derivative}, we have
\begin{equation*}
	\partial_{\varphi}A_{\varphi} (v,v) = A_{1,\varphi}(v,v) = - Q_{\varphi}(v),
\end{equation*}
and therefore
\begin{equation*}
	(\varphi,v) \mapsto  Q_{\varphi}(v), \quad \D{q}\times\HH{q} \to \HH{q-r}
\end{equation*}
is smooth.

(c) The set $\mathrm{Isom}(\HH{q},\HH{q-r})$ is open in $\mathcal{L}(\HH{q},\HH{q-r})$ and the mapping
\begin{equation*}
    P \mapsto P^{-1}, \quad \mathrm{Isom}(\HH{q},\HH{q-r}) \to \mathrm{Isom}(\HH{q-r},\HH{q})
\end{equation*}
is smooth (it is even real analytic). Besides $A_{\varphi} \in \mathrm{Isom}(\HH{q},\HH{q-r})$, for all $\varphi \in \D{q}$, and the mapping
\begin{equation*}
    \varphi \mapsto A_{\varphi}, \quad \D{q} \to \mathrm{Isom}(\HH{q},\HH{q-r})
\end{equation*}
is smooth. Thus
\begin{equation*}
	(\varphi,w) \mapsto  A^{-1}_{\varphi}(w), \quad  \HH{q-r} \to \HH{q}
\end{equation*}
is smooth.
\end{proof}

We have, in particular, the following corollary.

\begin{cor}{\emph(Smoothness of the $H^{s}$ metric and its spray)}
Let $s \ge 1/2$ and assume that $q > 3/2$, $q-2s \ge 0$. Then the right-invariant, weak Riemannian metric defined on $\DiffS$ by the inertia operator $A=\Lambda^{2s}$ extends to a \emph{smooth} weak Riemannian metric on the Banach manifold $\D{q}$ with a smooth geodesic spray.
\end{cor}

\begin{rem}
When $A$ is a \emph{differential operator}, theorem~\ref{thm:smoothness_spray} can be sharpened; in that case, we can conclude that the spray is smooth when $q > 3/2$ and $q \ge r-1$. Indeed, we have
\begin{equation*}
   \int_{\Circle} \left\{ [A,u] u_{x} - 2(A u)u_{x} \right\} \, dx = 0,
\end{equation*}
so that $[A,u] u_{x} - 2(A u)u_{x} = DB(u)$ where $B$ is a quadratic, differential operator of order $r-1$ (see~\cite[Chapter 4]{Olv1993}). Hence
\begin{equation*}
  (\varphi,v) \mapsto B_{\varphi}(v), \quad \D{q}\times \HH{q} \to \HH{q-r+1}
\end{equation*}
is smooth for $q \ge r-1$. Moreover, the symbol $a$ of $A$ is a real, even polynomial (because $A$ is $L^{2}$-symmetric) with no real roots (because $A$ is invertible). Therefore, $A$ can be written as
\begin{equation*}
  A = A^{\prime}(\alpha -iD)(\bar{\alpha} -iD),
\end{equation*}
where $\alpha\in \CC\setminus\RR$ and $A^{\prime}$ is an invertible, differential operator of degree $r-2$. We have thus
\begin{equation*}
  A^{-1}D = (A^{\prime})^{-1}\frac{1}{2}\left( \frac{\alpha}{\mathrm{Im}\,\alpha}(\alpha -iD)^{-1} - \frac{\bar{\alpha}}{\mathrm{Im}\,\alpha}(\bar{\alpha} -iD)^{-1}\right).
\end{equation*}
The conclusion follows now from the fact that
\begin{equation*}
  \varphi \mapsto A^{\prime}_{\varphi}, \quad \D{q} \to \mathrm{Isom}(\HH{q},\HH{q-r+2})
\end{equation*}
is smooth if $q > 3/2$ and $q \ge r-2$ and that
\begin{equation*}
  \varphi \mapsto (\alpha -iD)_{\varphi},\,(\bar{\alpha} -iD)_{\varphi} \quad \D{q} \to \mathrm{Isom}(\HH{q},\HH{q-1})
\end{equation*}
are smooth if $q > 3/2$. This applies in particular for the Camassa-Holm equation where the spray is smooth for $q > \frac{3}{2}$ (in \cite{Mis2002}, it was proved that the spray is of class $C^{1}$ for $q > \frac{3}{2}$).
\end{rem}

\section{Well-posedness}
\label{sec:existence_results}

In this section, we will prove local existence and uniqueness of the initial value problem for the geodesics of the right-invariant $H^{s}$ metric on the Fr\'{e}chet-Lie group $\DiffS$, and more generally for any right-invariant weak Riemannian metric for which the inertia operator $A$ is such that
\begin{equation}\label{eq:hypo_A}
	\varphi \mapsto A_{\varphi} = R_{\varphi} \circ A \circ R_{\varphi^{-1}}, \quad \D{q} \to \mathrm{Isom}(\HH{q},\HH{q-r})
\end{equation}
is smooth, for $q > 3/2$ and $q-r \ge 0$. Under these assumptions, the metric admits a smooth spray $F_{q}$ defined on $T\D{q}$ (c.f. theorem~\ref{thm:smoothness_spray}) and we can apply the Picard-Lindel\"{o}f theorem. For each $(\varphi_{0},v_{0})\in T\D{q}$, there exists a \emph{unique non-extendable} solution
\begin{equation*}
    (\varphi,v)\in C^\infty(J_{q}(\varphi_{0},v_{0}),T\D{q}),
\end{equation*}
of the Cauchy problem
\begin{equation}\label{eq:CauchyP}
\left\{\begin{aligned}
    \varphi_{t} &= v, \\
    v_{t} &= S_{\varphi}(v),
\end{aligned}
\right.
\end{equation}
with $\varphi(0)=\varphi_{0}$ and $v(0)=v_{0}$, defined on some \emph{maximal interval of existence} $J_{q}(\varphi_{0},v_{0})$, which is open and contains $0$. Note that in general $J_{q}(\varphi_{0},v_{0})\ne \mathbb{R}$, meaning that the solutions are not \emph{global}.

To prove well-posedness of the Cauchy problem \eqref{eq:CauchyP} on the smooth manifold $T\DiffS$, we need precise regularity properties of solutions to \eqref{eq:CauchyP} on each Hilbert approximation manifold $T\D{q}$. More precisely, assume that $(\varphi_{0},v_{0})\in T\D{q+1}$. Then, we may solve~\eqref{eq:CauchyP} in $T\D{q}$ and in $T\D{q+1}$. Since solutions on each level are non-extendable, we clearly have
\begin{equation}\label{eq:triv_incl}
    J_{q+1}(\varphi_{0},v_{0})\subset J_{q}(\varphi_{0},v_{0}),
\end{equation}
which could lead to
\begin{equation*}
  \bigcap_{q} J_{q}(\varphi_{0},v_{0})=\{0\}.
\end{equation*}

The remarkable observation that the maximal interval of existence is independent of the parameter $q$, due to the right-invariance of the spray (cf. lemma~\ref{lem:noloss}) was pointed out in \cite[Theorem 12.1]{EM1970}. This makes it possible to avoid Nash-Moser type schemes to prove local existence of smooth geodesics.

\begin{lem}[No loss, nor Gain]\label{lem:noloss}
Given $(\varphi_{0},v_{0})\in T\D{q+1}$, we have
\begin{equation*}
  J_{q+1}(\varphi_{0},v_{0})= J_{q}(\varphi_{0},v_{0}),
\end{equation*}
for $q>(3/2)$ and $q-r \ge 0$.
\end{lem}

\begin{proof}
Let $\Phi_{q}$ be the flow of the spray $F_{q}$ and $R_{s}$ be the (right) action of the rotation group $\Circle$ on $\D{q}$, defined by
\begin{equation*}
    \big(R_{s} \cdot \varphi \big)(x): = \varphi(x+s),\quad \varphi \in \D{q},\quad x\in\Circle.
\end{equation*}
This action induces an action on $T\D{q}$ given by
\begin{equation*}
    \big(R_{s} \cdot (\varphi,v)\big)(x): = (\varphi(x+s),v(x+s)),\quad (\varphi,v)\in T\D{q},\quad x\in\Circle.
\end{equation*}
Note that if $(\varphi,v)\in T\D{q+1}$, then\footnote{We will avoid to write $TR_{s}$, $T(TR_{s})$, \dots and simply keep the notation $R_{s}$.}
\begin{equation*}
  s \mapsto R_{s} \cdot (\varphi,v), \quad \Circle \to T\D{q}
\end{equation*}
is a $C^{1}$ map, and that
\begin{equation*}
    \frac{d}{ds} R_{s}\cdot (\varphi,v) = (\varphi_{x},v_{x}).
\end{equation*}
Therefore, if $(\varphi_{0},v_{0})\in T\D{q+1}$, we get
\begin{equation}\label{eq:D_Lie_derivative}
    \left.\frac{d}{ds}\right\vert_{s=0} \Phi_{q}(t, R_{s} \cdot (\varphi_{0},v_{0})) = \partial_{(\varphi,v)}\Phi_{q}(t,(\varphi_{0},v_{0})).(\varphi_{0,x},v_{0,x}).
\end{equation}
On the other hand, the spray $F_{q}$ is invariant under each right-invariant translation $R_{\varphi}$ where $\varphi\in\D{q}$. The same is true for its flow $\Phi_{q}$, and hence
\begin{equation*}
    \Phi_{q}(t,R_{s} \cdot (\varphi_{0},v_{0})) = R_{s} \cdot \Phi_{q}(t,(\varphi_{0},v_{0}))\quad\text{for all}\quad t\in J_{q}(\varphi_{0},v_{0}),\ s\in\RR.
\end{equation*}
We get thus
\begin{equation*}
    \partial_{(\varphi,v)}\Phi_{q}(t,(\varphi_{0},v_{0})).(\varphi_{0,x},v_{0,x}) = (\varphi_{x}(t),v_{x}(t)).
\end{equation*}
But $\partial_{(\varphi,v)}\Phi_{q}(t,(\varphi_{0},v_{0})).(\varphi_{0,x},v_{0,x})$ belongs to $\HH{q}\times\HH{q}$, and hence
\begin{equation*}
    (\varphi(t),v(t))\in T\D{q+1}\quad\text{for all}\quad t\in J_{q}(\varphi_{0},v_{0}).
\end{equation*}
We conclude therefore that
\begin{equation*}
    J_{q}(\varphi_{0},v_{0}) = J_{q+1}(\varphi_{0},v_{0}),
\end{equation*}
which completes the proof.
\end{proof}

\begin{rem}\label{rem:noloss_nogain}
Lemma~\ref{lem:noloss} states that there is no loss of spatial regularity during the evolution of \eqref{eq:CauchyP}. By reversing the time direction, it follows from the unique solvability that there is also no gain of regularity in the following sense: Let $(\varphi_{0},v_{0})\in T\D{q}$ be given and assume that $(\varphi(t_{1}),v(t_{1}))\in T\D{q+1}$ for some
$t_{1}\in J_{q}(\varphi_{0},v_{0})$. Then $(\varphi_{0},v_{0})\in T\D{q+1}$.
\end{rem}

We get therefore the following local existence result.

\begin{thm}\label{thm:smooth_flow}
Let \eqref{eq:hypo_A} be satisfied and consider the geodesic flow on the tangent bundle $T\DiffS$ induced by the inertia operator $A$. Then, given any $(\varphi_{0},v_{0})\in T\DiffS$, there exists a unique non-extendable solution
\begin{equation*}
  (\varphi, v)\in C^\infty(J,T\DiffS)
\end{equation*}
of \eqref{eq:CauchyP} on the maximal interval of existence $J$, which is open and contains $0$.
\end{thm}

And we obtain well-posedness of the Euler equation.

\begin{cor}\label{cor:Euler-well-posedness}
The corresponding Euler equation has for any initial data $u_{0}\in\CS$ a unique non-extendable smooth solution
\begin{equation*}
  u\in C^\infty(J,\CS).
\end{equation*}
The maximal interval of existence $J$ is open and contains $0$.
\end{cor}

It is known that the Euler equation induced by the inertia operator
\begin{equation*}
  A=\op{1+k^2}
\end{equation*}
leads to the classical periodic Camassa-Holm equation
\begin{equation}\label{eq:CH}
  u_{t} - u_{txx} + 3uu_{x} = 2u_{x}u_{xx} + uu_{xxx} ,\quad t > 0,\  x\in \Circle,
\end{equation}
cf. \cite{CK2003}. It may be interesting to briefly discuss another possible option for $A$, namely
\begin{equation*}
  A=\op{\abs{k}^{r} + \delta_{0}(k)}.
\end{equation*}
Observe that Theorem~\ref{thm:main_theorem} is applicable provided $r\ge 1$. In that case, the mapping
\begin{equation*}
	\varphi \mapsto R_{\varphi}\circ A\circ R_{\varphi^{-1}}, \quad \D{q} \to \mathcal{L}(\HH{q},\HH{q-r})
\end{equation*}
is smooth for $q>(3/2)$ and $q-r \ge 0$. Since in addition, $A$ is a topological linear isomorphism from $\HH{q}$ onto $\HH{q-r}$, the operator $A$ satisfies clearly assumption \eqref{eq:hypo_A} and thus Theorem~\ref{thm:smooth_flow} guarantees the well-possedness, in the smooth category, of the corresponding Euler equation
\begin{equation}\label{eq:frac_euler}
    m_{t} + um_{x} + 2u_xm = 0,\quad m = \mu(u) + (-\Delta)^{r/2}u
\end{equation}
where $(-\Delta)^{r/2}:=\op{\abs{k}^{r}}$ and $\mu(u):=\int_{\Circle} u$. Note that $\int_{\Circle} m\,dx$ is a conserved quantity for the evolution under \eqref{eq:frac_euler}, since $\int_{\Circle} m_{t}\,dx=0$. Equation \eqref{eq:frac_euler} is of particular interest for
the values $r=2$ and $r=1$, respectively. In the first case we get the so-called \emph{$\mu$-Hunter-Saxton equation}, cf. \cite{LMT2010,EW2012}
\begin{equation}\label{eq:mu_HS}
  u_{txx} + uu_{xxx} + 2u_{x}u_{xx} - 2\mu(u)u_{x} = 0  ,\quad t > 0,\  x\in \Circle,
\end{equation}
In the case $r=1$ we get the so-called \emph{generalized CLM equation}, cf. \cite{EW2012}
\begin{equation}\label{eq:mu_CLM}
    Hu_{tx}+ uHu_{xx} + 2\mu(u)u_{x}+  2u_{x}Hu_{x}=0 ,\quad t > 0,\  x\in \Circle,
\end{equation}
where $H=\op{i \sgn(k)}$ denotes the Hilbert transform, acting on the spatial variable $x\in\Circle$. Note that $\op{\abs{k}}=H\circ D=(-\Delta)^{1/2}$.

\section{Exponential map and minimisation problems}
\label{sec:exponential_map}

The geodesic flow $\Phi_{q}$ on the Hilbert manifold $T\D{q}$ satisfies the following remarkable property
\begin{equation*}
 \Phi_{q}(t, \varphi, \sigma v) = \Phi_{q}(\sigma t, \varphi, v),\quad \sigma >0,
\end{equation*}
which is a consequence of the quadratic nature of the geodesic spray \cite[Chapter 4]{Lan1999}. Therefore, the time one map of the flow is defined on some open set $W_{q}$ of $T\D{q}$. The \emph{exponential map} $\riem_{q}$ is defined as
\begin{equation*}
  \riem_{q}: (\varphi,v) \mapsto \pi \circ \Phi_{q}(1,\varphi,v), \quad W_{q}: \to \D{q},
\end{equation*}
where $\pi: T\D{q} \to \D{q}$ is the canonical projection. For each $\varphi \in \D{q}$, we denote by $\riem_{q,\varphi}$, the restriction of $\riem_{q}$ to the tangent space $T_{\varphi}\D{q}$. Thus
\begin{equation*}
  \riem_{q,\varphi}: T_{\varphi}\D{q} \to \D{q}.
\end{equation*}
If the spray $F_{q}$ is smooth, then $\riem_{q, \id}$ is a local diffeomorphism from a neighbourhood $V_{q}$ of the $0 \in T_{\id}\D{q}$ onto a neighbourhood $U_{q}$ of $\id \in \D{q}$ \cite[Theorem 4.1]{Lan1999}.

This last assertion is in general no longer true on a \emph{Fr\'{e}chet manifold} and in particular on $\DiffS$. One may find useful to recall on this occasion that the \emph{group exponential} of $\DiffS$ is not a local diffeomorphism \cite{Mil1984}. Moreover, the Riemannian exponential map for the $L^{2}$ metric (Burgers equation) on $\DiffS$ is not a local $C^{1}$-diffeomorphism near the origin \cite{CK2002}. Nevertheless, it has been established in \cite{CK2002}, that for the Camassa-Holm equation -- which corresponds to the Euler equation of the $H^{1}$ metric on $\DiffS$ -- and more generally for $H^k$ metrics ($k\ge 1$) (see \cite{CK2003}), the Riemannian exponential map was in fact a \emph{smooth local diffeomorphism}. This result is still true for $H^{s}$ right-invariant metrics on $\DiffS$ provided $s \in [1/2, + \infty)$.

\begin{thm}\label{thm:exponential_map}
The exponential mapping $\riem_{\id}$ for the $H^{s}$-metric on $\DiffS$ is a smooth local diffeomorphism from a neighbourhood $V$ of $0$ in $\VectS$ onto a neighbourhood $U$ of $\id$ in $\DiffS$ for each $s\ge 1/2$.
\end{thm}

The proof of theorem~\ref{thm:exponential_map} relies mainly on a linearized version of the \emph{no loss, no gain lemma}~\ref{lem:noloss}, and is stated below. The full proof is similar to the one given for \cite[Theorem 14]{EK2011} and will be omitted.

\begin{lem}\label{lem:exponential_regularity}
Let $v \in V_{q}\cap \HH{q+1}$, we have
\begin{equation*}
    T_{v} \riem_{q,\id} \left(\HH{q+1}\right) = \HH{q+1}.
\end{equation*}
\end{lem}

\begin{proof}
Let $G_{q} : W_{q} \to \D{q}\times \D{q}$ be the smooth mapping defined by
\begin{equation*}
  G_{q}(\varphi, v) := (\varphi, \riem_{q}(\varphi, v)).
\end{equation*}
From the invariance of the flow $\Phi_{q}^{t}$ under the right action of $\DiffS$, we deduce that
\begin{equation*}
  G_{q}(R_{s}\cdot(\varphi, v)) = R_{s}\cdot G_{q}(\varphi, v)
\end{equation*}
where $R_{s}$ denotes the natural action on derived spaces, induced by the action of the rotation group $\Circle$ on $\D{q}$
\begin{equation*}
    \big(R_{s}\cdot \varphi \big)(x): = \varphi(x+s).
\end{equation*}
Therefore, we have
\begin{equation}\label{eq:linearized-equivariance}
  TG_{q}.(R_{s}\cdot(\varphi, v, \delta \varphi,\delta v)) = R_{s}\cdot TG_{q}.(\varphi, v, \delta \varphi,\delta v).
\end{equation}
Now, if $(\varphi, v, \delta \varphi,\delta v) \in T^{2}\D{q+1}$, then
\begin{equation*}
  s \mapsto R_{s}\cdot (\varphi, v, \delta \varphi,\delta v), \quad \Circle \to T^{2}\D{q}
\end{equation*}
is a $C^{1}$ mapping and
\begin{equation*}
    \frac{d}{ds} R_{s}\cdot(\varphi, v, \delta \varphi,\delta v) = (\varphi_{x}, v_{x}, \delta \varphi_{x},\delta v_{x}).
\end{equation*}
Therefore, taking derivatives in $s$, at $s=0$, in equation~\eqref{eq:linearized-equivariance}, we get
\begin{multline*}
  T^{2}G_{q}.(\varphi, v, \delta \varphi,\delta v, \varphi_{x}, v_{x}, \delta \varphi_{x},\delta v_{x}) \\
  = \left(\varphi, v, \delta \varphi,\delta v, \varphi_{x}, \big(\riem_{q}(\varphi,v)\big)_{x}, \delta \varphi_{x},\big(T\riem_{q}. (\varphi, v, \delta \varphi,\delta v)\big)_{x}\right).
\end{multline*}
Since the left hand-side of the preceding equation lies in $T^{2}(\D{q}\times\D{q})$, we deduce that
\begin{equation*}
  T_{v}\riem_{q,\varphi}\left(\HH{q+1}\right) \subset \HH{q+1}
\end{equation*}
as soon as $(\varphi, v) \in \D{q+1} \times \HH{q+1}$. Since a similar statement can be made for $\left(T_{v}\riem_{q,\varphi}\right)^{-1}$, the proof of lemma~\ref{lem:exponential_regularity} is complete.
\end{proof}

\begin{rem}
Let $U$ and $V$ be the neighbourhoods introduced in theorem~\ref{thm:exponential_map}. We define
\begin{equation*}
  \mathcal{V} := \bigcup_{\varphi \in \DiffS} R_{\varphi}V,
\end{equation*}
which is an open neighbourhood of the \emph{zero section} in $T\DiffS$, and
\begin{equation*}
  \mathcal{U} := \set{(\varphi,\psi) \in \DiffS \times \DiffS;\; \psi\circ\varphi^{-1} \in U},
\end{equation*}
which is an open neighbourhood of the \emph{diagonal} in $\DiffS \times \DiffS$. One can deduce, from theorem~\ref{thm:exponential_map}, that the mapping
\begin{equation*}
  G : \mathcal{V} \to \mathcal{U}, \quad (\varphi,v) \mapsto (\varphi,\riem_{\varphi}(v))
\end{equation*}
is a smooth diffeomorphism.
\end{rem}

The restriction of $\riem_{q,\varphi}$ to $V_{q}$ defines a local chart around $\id$ in $\D{q}$. On this chart, called a \emph{normal chart}, we have \emph{local polar coordinates}, defined as follows. Given $\varphi \in U_{q} - \set{\id}$, there is a $v\in V_{q}\setminus\{0\}$ such that $\varphi = \riem_{q,\varphi}(v)$. Letting now
\begin{equation*}
    w := v/\norm{v}_{H^{s}}, \qquad \rho := \norm{v}_{H^{s}},
\end{equation*}
we have that $\varphi=\riem(\rho w)$ and $(\rho,w)$ are called the {\emph{polar coordinates}} of $\varphi\in U- \set{\id}$. Note that $(\rho,w)$ depend smoothly of $\varphi$ and that $\rho(\varphi) \to 0$ as $\varphi\to \id$.

As can be checked in \cite{Lan1999}, the following result is valid not only for a strong Riemannian metric but also for a \emph{weak} Riemannian metric, \emph{provided} there exists a compatible, symmetric covariant derivative.

\begin{lem}\label{lem:Lang_inequality}
 For a piecewise $C^{1}$ curve $\varphi : [a,b] \to U_{q} - \set{\id}$, we have the inequality
 \begin{equation}\label{eq:arc-length-bound}
    L_{s}(\gamma) \ge \abs{\rho(b) - \rho(a)},
 \end{equation}
 where
 \begin{equation*}
   L_{s}(\gamma) := \int_{a}^{b} \norm{R_{\varphi^{-1}}\varphi_{t}}_{H^{s}} \,dt .
 \end{equation*}
\end{lem}

A consequence of lemma~\ref{lem:Lang_inequality} is that the length of any path which \emph{lies inside} the normal neighbourhood is bounded below by $r := \abs{\rho(b) - \rho(a)}$. Note also that a path, of constant velocity norm, which minimizes locally the arc-length minimizes also the \emph{energy} defined as
\begin{equation*}
  E_{s}(\gamma) := \frac{1}{2}\int_{a}^{b} \norm{R_{\varphi^{-1}}\varphi_{t}}^{2}_{H^{s}} \,dt .
\end{equation*}
We get therefore the following theorem.

\begin{thm}
  Let $s \ge 1/2$. Given two nearby diffeomorphisms $\varphi, \psi \in \DiffS$, there exists a unique geodesic for the right-invariant $H^{s}$ metric on $\DiffS$, joining them and which minimizes \emph{locally} the \emph{arc-length} and the \emph{energy}.
\end{thm}

On a \emph{strong Riemannian manifold}, given two nearby points, there exists a unique geodesic, joining these two points, which minimizes \emph{globally} the arc-length and the energy. Note however, that for a \emph{weak metric}, this might not be true. Indeed, in lemma~\ref{lem:Lang_inequality}, the bound \eqref{eq:arc-length-bound} might not be true for a path which leaves the normal neighbourhood \emph{before} leaving the (weak ball) of radius $r$ defined as
\begin{equation*}
    B_{s}(\id,r) := \set{\varphi \in U;\; \rho(\varphi) \le r}.
\end{equation*}
This happens, in particular, for the critical exponent $s=1/2$ as it follows from \cite{BBHM2013}. Note also that, for the $L^{2}$ metric, the situation is even worse since the energy has no local minimum~\cite{Bru2013}. The problem wether the geodesic joining two nearby points is a \emph{global minimum} for $s > 1/2$ seems to be still an open problem.

To make this clear we close this section by a remark concerning the geodesic semi-distance $d_{s}$ induced by the $H^{s}$ metric and defined as the greatest lower bound of path-lengths $L_{s}(\varphi)$, for piecewise $C^{1}$ paths $\varphi(t)$ in $\DiffS$ joining $\varphi_{0}$ and $\varphi_{1}$. It was first shown in \cite{MM2005}, that this semi-distance vanishes identically for the $L^{2}$ right-invariant metric on the diffeomorphism group of any compact manifold. More recently, it was shown in \cite{BBHM2013} that $d_s$ vanishes identically on $\DiffS$ if $s\in[0,1/2]$, whereas $d_s$ is a distance for $s > 1/2$
\begin{equation*}
    \forall \varphi_{0},\varphi_{1} \in \DiffS, \quad \varphi_{0} \ne \varphi_{1} \; \Rightarrow \; d_{s}(\varphi_{0},\varphi_{1}) >0.
\end{equation*}
Anyway, it should be noted that lemma~\ref{lem:Lang_inequality} does not imply that the geodesic semi-distance is in fact a distance.

\section{Euler equations on homogeneous spaces}
\label{sec:homogeneous_spaces}

The theory of Euler equations on a homogeneous space $G/K$ has been developed in \cite{KM2003}. In that case, the geodesic flow for a right-invariant Riemannian metric on the homogeneous space $G/K$, can be reduced to the so called \emph{Euler-Poincar\'{e}} equation
\begin{equation}\label{eq:Euler-Poincare}
    m_{t} = ad^{*}_{u}\,m, \qquad m \in \gstar,
\end{equation}
on the dual space $\gstar$ of the Lie algebra of $G$ (see \cite{KM2003} or Poincar\'{e}'s original paper~\cite{Poi1901}). Unfortunately, there is no natural \emph{contravariant} formulation of equation~\eqref{eq:Euler-Poincare}, leading to an \emph{Euler equation}, as it is the case for a Lie group. In that case, the Eulerian velocity (defined using a lift $g(t)$ in $G$ of a path $x(t)$ in $G/K$) is only defined \emph{up to a path} in $K$ and the relation between $u$ and $m$ is not one-to-one (see \cite{TV2011} for a recent survey on this subject). Another way to treat the problem is to introduce \emph{sub-Riemannian geometry} on $G$ (see \cite{GMV2012,GMV2012a} for a deep study of this approach for $\DiffS$).

These difficulties clear away if $K$ is a normal subgroup. Indeed, in that case, the coset manifold $G/K$ is a Lie group equipped with a right-invariant Riemannian metric. But this special case is not very useful for our study, since $\DiffS$ is a \emph{simple} group: it has no nontrivial normal subgroups (see \cite{GR2007}).

Hopefully, there is another situation where a right-invariant Riemannian metric on a homogeneous space can be reduced to the ordinary theory of the Euler equation on a Lie group, namely when there exists a section of the projection map $\pi : G \to G/K$ onto a subgroup $H$ of $G$. This situation occurs for $\DiffS/\RotS$, in which case $H := \DiffSfix{1}$, the subgroup of diffeomorphisms which fix one point, and for $\DiffS/\PSL(2,\RR)$ in which case $H := \DiffSfix{3}$, the subgroup diffeomorphisms which fix three points. More precisely, we have the following.

\begin{lem}\label{lem:simple_transitive_action}
Let $G$ be a group and $H,K$ some subgroups of $G$. Suppose that
\begin{enumerate}
  \item The restriction to $H$ of the projection map $\pi : G \to G/K$ is surjective,
  \item $H \cap K = \set{e}$.
\end{enumerate}
Then $H$ acts \emph{simply} and \emph{transitively} on $G/K$.
\end{lem}

\begin{rem}
As a result, if the hypothesis of lemma~\ref{lem:simple_transitive_action} are satisfied, then $G/K$ inherits a group structure. Note, however that the restriction of the projection $\pi : H \to G/K$ is a group morphism, if and only if, $K$ is a normal subgroup of $G$.
\end{rem}

\begin{proof}
By definition, the projection map $\pi$ sends and element $g\in G$ to the coset $Kg$. To show that the (right) action of $H$ on $G/K$ is transitive, it suffices to show that for any coset $Kg$ we can find $h \in H$ such that $Kh=Kg$. But this means precisely that $\pi : H \to G/K$ is surjective. Hence the transitivity of the action is equivalent to the surjectivity of $\pi$. To prove that the action is simple, it is enough to show that the only element $h\in H$ which fixes the coset $K$ is $h=e$, the unit element. But this means $Kh = K$, and thus $h \in K \cap H$, which leads to $h=e$ by condition (2). Note that this implies that $\pi : H \to G/K$  is injective.
\end{proof}

The mentioned scenario is summarized in the following proposition.

\begin{prop}\label{prop:homogeneous_riemannian_structure}
Suppose that $H$ and $K$ are closed subgroups of a Lie group $G$ such that $H\cap K = \set{e}$ and such that $\pi : H\to G/K$ is surjective. Let $\g$, $\mathfrak{h}$, and $\mathfrak{k}$ denote the Lie algebras of $G$, $H$ and $K$, respectively. Let $A : \g \to \gstar$ be the inertia operator of a (degenerate) non-negative inner product $\langle \cdot, \cdot \rangle$ on $\g$, which satisfies the following conditions:
\begin{equation}\label{eq:condition-kernel}
  \ker A = \mathfrak{k},
\end{equation}
and
\begin{equation}\label{eq:condition-covariance}
  \langle \Ad_{k}u, \Ad_{k}v \rangle = \langle u,v \rangle, \quad \forall k \in K, \, \forall u,v \in \g ,
\end{equation}
where $\Ad$ is the adjoint action of $G$ on its Lie algebra $\g$. Then, $A$ induces a right-invariant Riemannian metric $\gamma$ on $G/K$ and $\pi: H \to G/K$ is a Riemannian isometry between $H$, endowed with the right-invariant metric induced by $A$, and $(G/K, \gamma)$.
\end{prop}

\begin{rem}
In the situation described, we have $\g = \mathfrak{k} \oplus \mathfrak{h}$ and $\mathfrak{h}^{*}$ can be identified with
\begin{equation*}
    \mathfrak{k}^{0} = \set{m \in \gstar ; \; (m,w) = 0, \, \forall w \in \mathfrak{k}}.
\end{equation*}
Now, condition~\eqref{eq:condition-covariance} implies\footnote{If the subgroup $K$ is connected, \eqref{eq:condition-covariance} and \eqref{eq:condition-covariance-infinitesimal} are equivalent.}
\begin{equation}\label{eq:condition-covariance-infinitesimal}
    \langle \ad_{w}u,v \rangle = -\langle u,\ad_{w}v \rangle, \quad \forall w \in \mathfrak{k}, \, \forall u,v \in \g.
\end{equation}
Therefore, we have
\begin{equation*}
    \left( \ad^{*}_{u} A(v), w \right) = - \left( \ad^{*}_{v} A(u), w \right), \quad \forall w \in \mathfrak{k}, \, \forall u,v \in \g,
\end{equation*}
where $\ad^{*}$ is the coadjoint action of $\g$ on $\gstar$, defined by
\begin{equation*}
  (\ad^{*}_{u}m, v) = -(m,[u,v]), \quad u,v \in \g, \, m \in \gstar.
\end{equation*}
Hence
\begin{equation*}
    \ad^{*}_{u}A(v) + \ad^{*}_{v}A(u) \in \mathfrak{k}^{0} = \range A
\end{equation*}
and Arnold's bilinear operator
\begin{equation*}
    B(u,v) = \frac{1}{2}A^{-1}\Big[ \ad^{*}_{u}A(v) + \ad^{*}_{v}A(u) \Big]
\end{equation*}
is well-defined as a mapping from $\mathfrak{h} \times \mathfrak{h}$ to $\mathfrak{h}$. The Euler equation on $\mathfrak{h}$ is given by
\begin{equation}\label{eq:homogeneous_Euler}
    u_{t} = - B(u,u) = - A^{-1}\Big[ \ad^{*}_{u}A(u) \Big].
\end{equation}
\end{rem}

In the next subsections, we will extend our main theorems to some geodesic equations on $\DiffS/\RotS$ and $\DiffS/\PSL(2,\RR)$. Since the proofs are very similar to what has been done so far, we will not exhibit all the details but only point out crucial changes.


\subsection{Euler equations on the coadjoint orbit $\DiffS/\RotS$}
\label{subsec:S1}

Let $\RotS$ denotes the subgroup of all rigid rotations of $\Circle$ and
\begin{equation*}
    \DiffS/\RotS,
\end{equation*}
be the corresponding homogeneous space of right cosets. Let $\DiffSfix{1}$ be the subgroup of $\DiffS$ consisting of all diffeomorphisms of $\Circle$ which fix one arbitrarily point (say $x_{0}$). It is easy to check that the conditions of lemma~\ref{lem:simple_transitive_action} are satisfied and hence that the canonical projection
\begin{equation*}
    \DiffSfix{1} \to \DiffS/\RotS
\end{equation*}
is a bijection. The group $\DiffSfix{1}$ is a Fr\'{e}chet Lie group and we can use the Fr\'{e}chet manifold structure of $\DiffSfix{1}$ to endow the quotient space $\DiffS/\RotS$ with a Fr\'{e}chet manifold structure, so that the canonical projection becomes a diffeomorphism. The Lie algebras of $\DiffSfix{1}$ and $\RotS$ are given by
\begin{equation*}
    \CSfix{1}:=\{u\in\CS \,;\, u(x_{0})=0 \} \quad \text{and} \quad \RR\cdot w_{0},
\end{equation*}
respectively, where $w_{0}$ stands for the constant function with value $1$.

$\DiffSfix{1}$ is an ILH space; a Hilbert approximation being given by the Hilbert manifolds
\begin{equation*}
    \Dfix{1}{q}:=\{\varphi\in \D{q}\,;\,\varphi(x_{0})=x_{0}\},
\end{equation*}
modelled on the Hilbert spaces
\begin{equation*}
    \HHfix{1}{q}:=\{u\in \HH{q}\,;\, u(x_{0})=0\}.
\end{equation*}
Note that $\Dfix{1}{q}$ is a closed submanifold of the Hilbert manifold $\D{q}$ and a closed topological subgroup of $\D{q}$, for $q > 3/2$.

Let $A=\op{p(k)}$ be a $L^{2}$-symmetric, Fourier multiplier on $\CS$ and assume that its symbol satisfies
\begin{equation*}
    p(k)=0\quad \text{iff}\; k=0,
\end{equation*}
which is equivalent to $\ker A=\RR\cdot w_{0}$. We have $\ad_{w_{0}} = -D$, so that condition~\eqref{eq:condition-covariance-infinitesimal} is satisfied. Since $\RotS$ is connected, hypotheses of proposition~\ref{prop:homogeneous_riemannian_structure} are fulfilled. If $A$ is of order $r\ge 1$, then $A$ extends to $\HHfix{1}{q}$, for all $q> 3/2$, and
\begin{equation*}
    A \in \mathrm{Isom}(\HHfix{1}{q},\HHhat{1}{q-r}),
\end{equation*}
where
\begin{equation*}
    \HHhat{1}{q-r}:=\set{m\in\HH{q-r}\,;\, \hat{m}(0)=0}.
\end{equation*}
Then, for each $\varphi \in \Dfix{1}{q}$, $A$ induces a positive inner product on each tangent space, $T_{\varphi}\Dfix{1}{q}$, with flat map
\begin{equation*}
    \tilde{A}_{\varphi} = \varphi_{x}A_{\varphi} \in \mathrm{Isom}(\HHfix{1}{q},\HHhat{1}{q-r}).
\end{equation*}
Note that $\tilde{A}(T\Dfix{1}{q}) = \Dfix{1}{q} \times \HHhat{1}{q-r}$ and the proof of theorem~\ref{thm:rot_theorem} is similar to that of theorem~\ref{thm:smoothness_spray} and will be omitted.

\begin{thm}\label{thm:rot_theorem}
Let $A=\op{p(k)}$ be a $L^{2}$-symmetric, non negative, Fourier multiplier of order $r\ge 1$, satisfying
\begin{equation}\label{symbol_rot}
    p(k)=0\ \Longleftrightarrow\  k=0.
\end{equation}
Assume, in addition, that
\begin{equation*}
	\varphi \mapsto A_{\varphi} = R_{\varphi} \circ A \circ R_{\varphi^{-1}}, \quad \D{q} \to \mathcal{L}(\HH{q},\HH{q-r})
\end{equation*}
is smooth for $q > 3/2$ and $q-r \ge 0$. Then, the induced right-invariant metric on $\Dfix{1}{q}$ is smooth and has a smooth spray. Moreover, given any $(\varphi_{0},v_{0})\in T\DiffSfix{1}$, there exists a unique non-extendable solution
\begin{equation*}
    (\varphi, v)\in C^\infty(J,T\DiffSfix{1})
\end{equation*}
of the Cauchy problem for the associated geodesic spray on the maximal interval of existence $J$.
\end{thm}

We briefly discuss two special instances, namely
\begin{equation*}
  A = \op{k^{2}} \quad \text{and} \quad A=\op{\abs{k}}.
\end{equation*}

In the first case, $A=\op{k^{2}}$, we get the periodic \emph{Hunter-Saxton equation} (HS), see \cite{HS1991,Yin2004,BC2005,Len2007},
\begin{equation}\label{eq:HS}
    u_{txx} + 2u_{x}u_{xx} + uu_{xxx}  = 0.
\end{equation}

When $A=\op{\abs{k}}$, we get the \emph{Constantin-Lax-Majda equation} (CLM), see \cite{CLM1985,Wun2010,EKW2012},
\begin{equation}\label{eq:CLM}
    \partial_{t}(Hu_{x})+uHu_{xx}+2u_{x}Hu_{x}=0,
\end{equation}
where $H=\op{i \sgn(k)}$ denotes the Hilbert transform, acting on the spatial variable $x\in\Circle$.

Clearly, both symbols $(k^{2})_{k\in\ZZ}$ and $(\vert k\vert)_{k\in\ZZ}$ satisfy \eqref{symbol_rot}. Moreover, they also fulfill the hypotheses of theorem~\ref{thm:main_theorem}, so that theorem~\ref{thm:rot_theorem} is applicable to both \eqref{eq:HS} and \eqref{eq:CLM}.


\subsection{Euler equations on the coadjoint orbit $\DiffS/\PSL(2,\RR)$}
\label{subsec:PSL2}

Let $\PSL(2,\RR)$ denotes the subgroup of all rigid M\"{o}bius transformations which preserves the circle $\Circle$ and let
\begin{equation*}
    \DiffS/\PSL(2,\RR),
\end{equation*}
be the corresponding homogeneous space of right cosets. Let $\DiffSfix{3}$ be the subgroup of $\DiffS$ consisting of all diffeomorphisms of $\Circle$ which fix 3 arbitrary distinct points (say $x_{0},x_{1},x_{2}$). Then $\PSL(2,\RR) \cap \DiffSfix{3} = \set{e}$ and the canonical projection
\begin{equation*}
    \DiffSfix{3} \to \DiffS/\PSL(2,\RR)
\end{equation*}
is a bijection. The group $\DiffSfix{3}$ is a Fr\'{e}chet Lie group and we can use this Fr\'{e}chet structure to endow the quotient space $\DiffS/\PSL(2,\RR)$ with a Fr\'{e}chet manifold structure. In that way, the canonical projection becomes a diffeomorphism. The Lie algebras of $\DiffSfix{3}$ is given by
\begin{equation*}
    \CSfix{3}:=\{u\in\CS \,;\, u(x_{0})=0,\, u(x_{1})=0,\, u(x_{2})=0 \},
\end{equation*}
whereas the Lie algebra $\psl(2,\RR) \subset \CS$ of $\PSL(2,\RR)$ is the 3-dimensional subalgebra of $\CS$, generated by
\begin{equation*}
    w_{0}(x) := 1, \quad w_{1}(x) := \cos(x), \quad w_{-1}(x) := \sin(x).
\end{equation*}

An ILH structure on $\DiffSfix{3}$ is given by the Hilbert manifolds
\begin{equation*}
    \Dfix{3}{q}:=\set{\varphi\in \D{q};\; \varphi(x_{0})=x_{0},\, \varphi(x_{1})=x_{1},\, \varphi(x_{2})=x_{2}},
\end{equation*}
modelled on the Hilbert spaces
\begin{equation*}
    \HHfix{3}{q}:=\{u\in \HH{q}\,;\, u(x_{0})=0,\, u(x_{1})=0,\, u(x_{2})=0\}.
\end{equation*}
Note that $\Dfix{3}{q}$ is a closed submanifold of the Hilbert manifold $\D{q}$ and a topological subgroup of $\D{q}$, for $q > 3/2$.

\begin{thm}\label{thm:PSL2_theorem}
Let $A=\op{p(k)}$ be a $L^{2}$-symmetric, non negative, Fourier multiplier of order $r\ge 1$, satisfying
\begin{equation}\label{symbol_PSL2}
    p(k) = 0 \Longleftrightarrow  k \in \set{-1,0,1},
\end{equation}
and
\begin{equation*}
    \langle \ad_{w}u,v \rangle = -\langle u,\ad_{w}v \rangle, \quad \forall w \in \psl(2,\RR), \, \forall u,v \in \CS.
\end{equation*}
Assume, in addition, that
\begin{equation*}
	\varphi \mapsto A_{\varphi} = R_{\varphi} \circ A \circ R_{\varphi^{-1}}, \quad \D{q} \to \mathcal{L}(\HH{q},\HH{q-r})
\end{equation*}
is smooth for $q > 3/2$ and $q-r \ge 0$. Then, the induced right-invariant metric on $\Dfix{3}{q}$ is smooth and has a smooth spray. Moreover, given any $(\varphi_{0},v_{0})\in T\DiffSfix{3}$, there exists a unique non-extendable solution
\begin{equation*}
    (\varphi, v)\in C^\infty(J,T\DiffSfix{3})
\end{equation*}
of the Cauchy problem for the associated geodesic spray on the maximal interval of existence $J$.
\end{thm}

\begin{proof}
The proof is similar to that of theorem~\ref{thm:smoothness_spray}, except for point (c). Note that $A$ extends to $\HHfix{3}{q}$, for all $q> 3/2$, and that
\begin{equation*}
    A \in \mathrm{Isom}(\HHfix{1}{q},\HHhat{3}{q-r}),
\end{equation*}
where
\begin{equation*}
    \HHhat{3}{q-r}:=\set{m\in\HH{q-r}\,;\, \hat{m}(0)=0,\, \hat{m}(1)=0,\, \hat{m}(-1)=0}.
\end{equation*}
Let
\begin{equation*}
	\tilde{A}: \Dfix{3}{q}\times \HHfix{3}{q} \to \Dfix{3}{q}\times \HH{q-r}, \quad (\varphi, v) \to  (\varphi, \varphi_{x}A_{\varphi}v).
\end{equation*}
Given $\varphi \in \Dfix{3}{q}$, we have
\begin{equation*}
    \tilde{A}_{\varphi}(\HHfix{3}{q}) = \set{m \in \HH{q-r};\; \varphi_{x}(m\circ\varphi) \in \HHhat{3}{q-r}},
\end{equation*}
so we cannot conclude immediately that $\tilde{A}(T\Dfix{3}{q})$ is a trivial bundle, as this was the case for $\tilde{A}(T\D{q})$ and $\tilde{A}(T\Dfix{1}{q})$. To overcome this difficulty, we first remark that $\tilde{A}$ is a vector bundle morphism. Moreover, the continuous linear map $\tilde{A}_{\varphi}: \HHfix{3}{q} \to \HH{q-r}$ is injective and \emph{splits}\footnote{If $E,F$ are Banach spaces and $\Lambda : E \to F$ is a continuous linear map, which is injective, then we say that $\Lambda$ \emph{splits} if $\Lambda(E)$ is \emph{closed} and \emph{complemented} in $F$ (i.e there exists a closed subspace $G$ of $F$ such that $F = \Lambda(E) \oplus G$). Note that if $F$ is a Hilbert space, then every injective, continuous linear map with closed range, splits.} because $\tilde{A}_{\varphi}(\HHfix{3}{q})$ is a closed subspace of $\HH{q-r}$, for every $\varphi \in \Dfix{3}{q}$. Then according to \cite[Proposition 3.1]{Lan1999}, $\tilde{A}(T\Dfix{3}{q})$ is a subbundle of $\Dfix{3}{q}\times \HH{q-r}$ which is isomorphic to $\Dfix{3}{q}\times \HHfix{3}{q-r}$. Finally, an argument similar to the one given in point (c) of the proof of theorem~\ref{thm:smoothness_spray} does apply and achieves the proof.
\end{proof}

An important application of Theorem~\ref{thm:PSL2_theorem} is the Euler-Weil-Petersson equation, which corresponds to the inertia operator
\begin{equation*}
    A := HD(D^{2}+1) = \op{\abs{k}(k^{2}-1)}.
\end{equation*}
This equation has been related with the Weil-Petersson metric on the universal Teichm\"{u}ller space $T(1)$ in \cite{NS1995,TT2006}. The corresponding geodesic flow has been extensively studied in \cite{Gay2009}. One of the main results of this paper is that the inertia operator $A$ defines on a suitable replacement\footnote{$\D{s}$, the space of homeomorphisms of class $H^{s}$ as well as their inverse is a topological group, if and only if, $s>3/2$.} for the ``diffeomorphism group of class $H^{3/2}$ '', a right-invariant \emph{strong Riemannian structure} which is, moreover, geodesically complete.

Our point of view is completely different because we are interested in the geodesic flow for the right-invariant metric on the Fr\'{e}chet Lie group $\DiffSfix{3}$ and its Hilbert approximation $\D{s}$ for $s >3/2$. The price to pay is the fact that the metric only defines a \emph{weak Riemannian structure}. Nevertheless, theorem~\ref{thm:PSL2_theorem} applies in this case. Indeed, $A$ satisfies the hypothesis of theorem~\ref{thm:main_theorem} and all conditions of theorem~\ref{thm:PSL2_theorem}. This proves short time existence of geodesics on $\DiffSfix{3}$, which doesn't seem to be a consequence of the results in \cite{Gay2009}.

\appendix

\section{Fourier multipliers}
\label{sec:fourier_multipliers}

In this Appendix, we recall and establish some basic results about \emph{Fourier multipliers}. Here and in the following, we use the notation
\begin{equation*}
	\e_{k}(x) = \exp(2\pi i k x),
\end{equation*}
for $k\in \ZZ$ and $x\in \Circle$.

\begin{lem}\label{lem:FOp}
Let $A$ be a \emph{continuous} linear operator on the Fr\'{e}chet space $\CSC$. Then the following three conditions are equivalent:
\begin{enumerate}
  \item $A$ commutes with all rotations $R_{s}$.
  \item $[A,D] = 0$, where $D = d/dx$.
  \item For each $k \in \ZZ$, there is a $a(k)\in\mathbb{C}$ such that $A\e_{k} = a(k)\e_{k}$.
\end{enumerate}
In that case, we say that $A$ is a \emph{Fourier multiplier}.
\end{lem}

Since every smooth function on the unit circle $\Circle$ can be represented by
its Fourier series, we get that
\begin{equation}\label{eq:Fourier_series}
    (Au)(x) = \sum_{k\in\ZZ} a(k)\hat u(k) \e_{k}(x),
\end{equation}
for every Fourier multiplier $A$ and every $u\in \CS$, where
\begin{equation*}
    \hat u(k):= \int_{\Circle} u(x) \overline{\e_{k}(x)}\,dx,
\end{equation*}
stands for the $k$-th Fourier coefficients of $u$.
The sequence $a : \ZZ\to\CC$ is called the \emph{symbol} of $A$. We use also the notation $A:=\op{a(k)}$ for the Fourier multiplier induced by the sequence $a$.

\begin{proof}
Given $s \in \RR$ and $u \in \CS$, let $u_{s}(x) := u(x+s)$. If $A$ commutes with translations we have
\begin{equation*}
    (Au)_{s}(x) = (Au_{s})(x).
\end{equation*}
Taking the derivative of both
sides
of this equation with respect to $s$ at $0$ and using the continuity of $A$, we get $DAu = ADu$ which proves the implication $(1) \Rightarrow (2)$.

If $[A,D] = 0$, then both $A\e_{k}$ and $\e_{k}$ are solutions of the linear differential equation $u^{\prime} = (-2\pi i k) u$ and are therefore equal, up to a multiplicative constant $a(k)$. This proves that $(2) \Rightarrow (3)$.

If $A\e_{k} = a(k)\e_{k}$, for each $k \in \ZZ$ and $A$ is continuous, then we have representation~\eqref{eq:Fourier_series}. Therefore
\begin{align*}
    (Au)_{s}(x) & = \sum_{k\in\ZZ} a(k)\hat u(k) \e_{k} (x + s) \\
        & = \sum_{k\in\ZZ} a(k)\widehat{u_{s}}(k) \e_{k}(x) = (Au_{s})(x),
\end{align*}
which proves that $(3) \Rightarrow (1)$.
\end{proof}

\begin{rem}\label{rem:L2_symmetry}
The space of Fourier multipliers is a \emph{commutative subalgebra} of the algebra of linear operators on $\CSC$. It contains all linear differential operators with constant coefficients. Note that a Fourier multiplier $A$ is $L^{2}$-symmetric iff its symbol $a$ is real.
\end{rem}

Let $I_{n}:=\set{1, \dotsc , n}$. Given a function $f$ and $n \ge 1$, we introduce
\begin{equation*}
  S_{f,n}(m_{0}, m_{1}, \dotsc , m_{n}) := m_{0} \sum_{p=0}^{n} (-1)^{p} \sum_{\substack{J \subset I_{n}, \\ \abs{J} = p}} f_{n}\big(m_{0} + \sum_{j \in J} m_{j}\big),
\end{equation*}
where $f_{n}(\xi) := \xi^{n-1}f(\xi)$ and $m_{j} \in \RR$, for $0 \le j \le n$.

\begin{lem}\label{lem:recursive-relation}
For each $n \ge 1$, we have
\begin{multline}\label{eq:recursive-relation}
  S_{f,n+1}(m_{0}, m_{1}, \dotsc , m_{n+1}) = \\
    \sum_{k=0}^{n} m_{k} \Big(  S_{f,n}(m_{0}, \dotsc , m_{n}) - S_{f,n}(m_{0}, \dotsc ,  m_{k} + m_{n+1}, \dotsc ,  m_{n})\Big).
\end{multline}
\end{lem}

\begin{proof}
Using the definition of $S_{f,n}$, the right hand-side of~\eqref{eq:recursive-relation} writes as
\begin{multline*}
  \sum_{p=0}^{n} (-1)^{p} \sum_{\substack{J \subset I_{n}, \\ \abs{J} = p}} \left\{ m_{0}^{2} f_{n}\big(m_{0} + \sum_{j \in J} m_{j}\big) - m_{0}(m_{0} + m_{n+1}) f_{n}\big(m_{0} + \sum_{j \in J} m_{j}\big) \right.\\
  + \left.\sum_{k=1}^{n} m_{0}m_{k} \Big[ f_{n}\big(m_{0} + \sum_{j \in J} m_{j}\big) - f_{n}\big(m_{0} + \sum_{j \in J} m_{j} + \delta_{J}(k) \, m_{n+1} \big)\Big] \right\},
\end{multline*}
which can be recast as
\begin{multline*}
  m_{0}\sum_{p=0}^{n} (-1)^{p} \sum_{\substack{J \subset I_{n}, \\ \abs{J} = p}} \left\{ \big(m_{0} + \sum_{j \in J} m_{j}\big) f_{n}\big(m_{0} + \sum_{j \in J} m_{j}\big) \right. \\
  - \left. \big(m_{0} + m_{n+1} + \sum_{j \in J} m_{j}\big) f_{n}\big(m_{0} + m_{n+1} + \sum_{j \in J} m_{j}\big)\right\},
\end{multline*}
and which is equal to
\begin{equation*}
  m_{0}\sum_{p=0}^{n} (-1)^{p} \sum_{\substack{J \subset I_{n}, \\ \abs{J} = p}} \left\{ f_{n+1}\big(m_{0} + \sum_{j \in J} m_{j}\big) - f_{n+1}\big(m_{0} + m_{n+1} + \sum_{j \in J} m_{j}\big) \right\},
\end{equation*}
since $f_{n+1}(\xi) = \xi f_{n}(\xi)$. Therefore, the right hand-side of~\eqref{eq:recursive-relation} is equal to
\begin{multline*}
    m_{0}\sum_{p=0}^{n} (-1)^{p} \sum_{\substack{J \subset I_{n+1}, \\ \abs{J} = p,\, n+1 \notin J} }
    f_{n+1}\big(m_{0} + \sum_{j \in J} m_{j}\big) \\
    + m_{0} \sum_{p=0}^{n} (-1)^{p+1} \sum_{\substack{J \subset I_{n+1}, \\ \abs{J} = p+1,\, n+1 \in J} } f_{n+1} \big(m_{0} + \sum_{j \in J} m_{j}\big),
\end{multline*}
which is exactly
\begin{equation*}
  S_{f,n+1}(m_{0}, m_{1}, \dotsc , m_{n+1}).
\end{equation*}
\end{proof}

\begin{lem}\label{lem:nth_derivative_symbol}
Let $A = \op{a(k)}$ be a Fourier multiplier on $\CS$, and let $(A_{n})$ be the sequence of multilinear operators defined inductively in lemma~\ref{lem:nth_derivative}. Then, for each $n \ge 1$, we have
\begin{equation}\label{eq:nth_derivative_symbol}
    A_{n}(\e_{m_{0}}, \dotsc ,\e_{m_{n}}) = a_{n}(m_{0}, m_{1}, \dotsc , m_{n}) \e_{m_{0} + m_{1} \dotsb + m_{n}},
\end{equation}
where
\begin{equation}\label{eq:sequence-an}
    a_{n}(m_{0}, m_{1}, \dotsc , m_{n}) = (2 \pi i )^{n} S_{a,n}(m_{0}, m_{1}, \dotsc , m_{n}).
\end{equation}
\end{lem}

\begin{rem}
For $n=1$, we have
\begin{equation*}
    a_{1}(m_{0}, m_{1}) = (2 \pi i) m_{0} \Big[ a(m_{0}) - a(m_{0} + m_{1}) \Big]
\end{equation*}
and for $n=2$, we get
\begin{multline*}
    a_{2}(m_{0}, m_{1}, m_{2}) = (2 \pi i)^{2} m_{0} \Big[ (m_{0} + m_{1} + m_{2}) a(m_{0} + m_{1} + m_{2}) \\
        - (m_{0} + m_{1}) a(m_{0} + m_{1}) - (m_{0} + m_{2}) a(m_{0} + m_{2}) + m_{0} a(m_{0}) \Big].
\end{multline*}
\end{rem}

\begin{proof}
From relation~\eqref{eq:recurrence_relation}, we obtain~\eqref{eq:nth_derivative_symbol} by induction, where the sequence $(a_{n})$ satisfies
\begin{multline}\label{eq:induction-an}
    a_{n+1}(m_{0}, \dotsc , m_{n+1}) = \\
    (2 \pi i) \sum_{k=0}^{n} m_{k} \Big[ a_{n}(m_{0}, \dotsc , m_{n}) - a_{n}(m_{0}, \dotsc ,  m_{k} + m_{n+1}, \dotsc ,  m_{n})\Big].
\end{multline}
For $n=1$, relation~\eqref{eq:sequence-an} is clear. Now, suppose inductively that \eqref{eq:sequence-an} holds for some $n\ge 1$. Then, using~\eqref{eq:induction-an}, we get
\begin{multline*}
  a_{n+1}(m_{0}, \dotsc , m_{n+1}) = \\
  (2 \pi i)^{n+1} \sum_{k=0}^{n} m_{k} \Big[ S_{a,n}(m_{0}, \dotsc , m_{n}) - S_{a,n}(m_{0}, \dotsc ,  m_{k} + m_{n+1}, \dotsc ,  m_{n})\Big],
\end{multline*}
which is equal to
\begin{equation*}
  (2 \pi i )^{n+1} S_{a,n+1}(m_{0}, m_{1}, \dotsc , m_{n+1}),
\end{equation*}
by virtue of lemma~\ref{lem:recursive-relation}. This achieves the proof.
\end{proof}

Recall that a function $f: \RR \to \RR$ is \emph{absolutely continuous} if $f$ has a derivative almost everywhere, the derivative is locally Lebesgue integrable and
\begin{equation*}
    f(b) = f(a) + \int_{a}^{b} f^{\prime}(\tau)\,d\tau ,
\end{equation*}
for all $a,b \in \RR$.

\begin{lem}\label{lem:n_order_estimate}
Let $f : \RR \to \RR$ and $n\ge 1$. Suppose that $f_{n}(\xi):= \xi^{n-1}f(\xi)$ is of class $C^{n-1}$, that $f_{n}^{(n-1)}$ is \emph{absolutely continuous} and that there exists $C_{n}>0$ and $r \ge 1$ such that
\begin{equation}\label{eq:fn_estimate}
    \abs{f_{n}^{(n)}(\xi)} \le C_{n} (1+\xi^{2})^{(r -1)/2},
\end{equation}
almost everywhere. Then
\begin{equation*}
    \abs{S_{f,n}(m_{0}, m_{1}, \dotsc , m_{n})} \le C_{n} \left(\prod_{j=0}^{n}\abs{m_{j}}\right) \sum_{J \subset I_{n}} \left(1 + \left( m_{0} + \sum_{j \in J} m_{j} \right)^{2}\right)^{(r -1)/2},
\end{equation*}
for all $m_{0}, m_{1}, \dotsc , m_{n}\in \RR$.
\end{lem}

\begin{proof}
Fix $n \ge 1$ and $m_{0}, m_{1}, \dotsc , m_{n} \in \RR$. Let $g_{k}$ be the sequence of functions defined inductively by
\begin{equation*}
    g_{0}(\xi) = f_{n}(\xi), \quad g_{k+1}(\xi) = g_{k}(\xi) - g_{k}(\xi + m_{n-k}),
\end{equation*}
for $k = 0, \dotsc , n-1$. We have in particular
\begin{equation*}
    S_{f,n}(m_{0}, m_{1}, \dotsc , m_{n}) = m_{0}\,g_{n}(m_{0}).
\end{equation*}
Let $K_{0} = \set{m_{0}}$ and for $p = 1, \dotsc , n$, let $K_{p}$ be the convex set generated by $K_{p-1}$ and $K_{p-1} + m_{p}$. Note that $K_{n}$ is the convex hull of the points $m_{0} + \sum_{j \in J} m_{j}$, for all subset $J$ of $\set{1, \dotsc , n}$. Let
\begin{equation*}
  M:= \max_{\xi \in K_{n}} (1+\xi^{2})^{(r -1)/2}.
\end{equation*}
By hypothesis, we have $\abs{g_{0}^{(n)}(\xi)} \le C_{n}M$ almost everywhere on $K_{n}$, and using the mean value theorem, we get inductively
\begin{equation*}
    \abs{g_{k}^{(n-k)}(\xi)} \le C_{n}M \abs{m_{n}} \dotsm \abs{m_{n-k+1}}, \qquad \forall \xi \in K_{n-k},
\end{equation*}
for $k = 1, \dotsc n$. In particular, we have
\begin{equation*}
    \abs{g_{n}(m_{0})} \le C_{n}M \prod_{j=1}^{n}\abs{m_{j}}.
\end{equation*}
Let's now estimate $M$. For $r \ge 1$, the function $\xi \mapsto (1+ \xi^{2})^{(r -1)/2}$ has no local maximum on $\RR$ (or is constant). Thus, it attains its maximum on the compact, convex set $K_{n}$ at some extremal point $m_{0} + \sum_{j \in J_{0}} m_{j}$ and we get
\begin{align*}
  M & = \left(1 + \left( m_{0} + \sum_{j \in J_{0}} m_{j} \right)^{2}\right)^{(r -1)/2} \\
    & \le \sum_{J \subset I_{n}} \left(1 + \left( m_{0} + \sum_{j \in J} m_{j} \right)^{2}\right)^{(r -1)/2},
\end{align*}
which achieves the proof.
\end{proof}

\begin{cor}\label{cor:nth_derivative_estimates}
Let $A = \op{a(k)}$ be a Fourier multiplier of order $r \ge 1$. Suppose that $a$ satisfies the hypothesis of lemma~\ref{lem:n_order_estimate}. Then, each $A_{n}$ extends to a bounded multilinear operator
\begin{equation*}
    A_{n} \in \mathcal{L}^{n+1}\left( \HH{q},\HH{q-r} \right)
\end{equation*}
when $q>3/2$ and $q - r \ge 0$.
\end{cor}

\begin{proof}
Given smooth functions $u_{0}, \dotsc ,u_{n}$, we have
\begin{equation*}
  \norm{A_{n}(u_{0}, \dotsc ,u_{n})}_{H^{q-r}}^{2} = \sum_{k\in \ZZ} (1+k^{2})^{q-r} \abs{\hat{A}_{n}(k)}^{2},
\end{equation*}
where
\begin{align*}
  \hat{A}_{n}(k) & := <A_{n}(u_{0}, \dotsc ,u_{n}),e_{k}>_{L^{2}} \\
        & \phantom{:} = \sum_{m_{0} + \dotsb + m_{n} = k}  a_{n}(m_{0}, \dotsc , m_{n}) \,  \widehat{u_{0}}(m_{0}) \dotsm \widehat{u_{n}}(m_{n}).
\end{align*}
If we assume, moreover, that $a$ satisfies the hypothesis of lemma~\ref{lem:n_order_estimate}, we get
\begin{equation*}
  \abs{\hat{A}_{n}(k)}^{2} \lesssim \sum_{J\subset I_{n}} \left(\hat{A}_{n,J}(k)\right)^{2}
\end{equation*}
where
\begin{multline*}
  \hat{A}_{n,J}(k) := \sum_{m_{0} + \dotsb + m_{n} = k} \Big(1+\big( m_{0} + \sum_{j \in J} m_{j}\big)^{2}\Big)^{(r-1)/2}
  \\
  \abs{\widehat{u_{0}^{\prime}}(m_{0})} \dotsm \abs{\widehat{u_{n}^{\prime}}(m_{n})}.
\end{multline*}
Recall now, that, for smooth functions on $\CS$, we have
\begin{equation*}
    \widehat{v_{0} \dotsm v_{n}}(k) = \sum_{m_{0} + \dotsb + m_{n} = k} \widehat{v_{0}}(m_{0}) \dotsm \widehat{v_{n}}(m_{n}).
\end{equation*}
Therefore, $\hat{A}_{n,J}(k)$ is the $k$-th Fourier coefficient of
\begin{equation*}
  \Lambda^{r-1}\left(\tilde{u}_{0}\prod_{j \in J}\tilde{u}_{j}\right)\prod_{j \in I_{n} \setminus J}\tilde{u}_{j},
\end{equation*}
where $\Lambda^{s} := \op{(1+k^{2})^{s/2}}$ and $\tilde{u}_{j}$ is the smooth function with Fourier coefficients
\begin{equation*}
  \widehat{\tilde{u}_{j}}(k) = \abs{\widehat{u_{j}^{\prime}}(k)}.
\end{equation*}
Thus
\begin{equation*}
  \norm{A_{n}(u_{0}, \dotsc ,u_{n})}_{H^{q-r}}^{2} \lesssim \sum_{J\subset I_{n}} \norm{\Lambda^{r-1}\left(\tilde{u}_{0}\prod_{j \in J}\tilde{u}_{j}\right)\prod_{j \in I_{n} \setminus J}\tilde{u}_{j}}_{H^{q-r}}^{2}.
\end{equation*}
Now, because $q-1 > 1/2$ and $0 \le q-r \le q-1$, we deduce from lemma~\ref{lem:pointwise_multiplication} that
\begin{align*}
  \norm{A_{n}(u_{0}, \dotsc ,u_{n})}_{H^{q-r}}^{2} & \lesssim \sum_{J\subset I_{n}} \norm{\Lambda^{r-1}\left(\tilde{u}_{0}\prod_{j \in J}\tilde{u}_{j}\right)}_{H^{q-r}}^{2} \norm{\prod_{j \in I_{n} \setminus J}\tilde{u}_{j}}_{H^{q-1}}^{2}
  \\
  & \lesssim \sum_{J\subset I_{n}} \norm{\tilde{u}_{0}\prod_{j \in J}\tilde{u}_{j}}_{H^{q-1}}^{2} \norm{\prod_{j \in I_{n} \setminus J}\tilde{u}_{j}}_{H^{q-1}}^{2}
  \\
  & \lesssim \norm{\tilde{u}_{0}}_{H^{q-1}}^{2} \dotsm \norm{\tilde{u}_{n}}_{H^{q-1}}^{2},
\end{align*}
because $\HH{q-1}$ is a multiplicative algebra. Since, moreover, $u_{j}^{\prime}$ and $\tilde{u}_{j}$ have the same $H^{q-1}$ norm, we obtain finally
\begin{equation*}
    \norm{A_{n}(u_{0}, \dotsc ,u_{n})}_{H^{q-r}} \lesssim \norm{u_{0}}_{H^{q}} \dotsm \norm{u_{n}}_{H^{q}},
\end{equation*}
which achieves the proof.
\end{proof}

\section{Boundedness properties of right translations}
\label{sec:boundedness_properties}

In this Appendix, we provide some explicit estimates for the right representation of $\D{q}$ on $\mathcal{L}(\HH{\rho},\HH{\rho})$, when $q >3/2$ and $0 \le \rho \le q$. These results are certainly not new but we give here very precise estimates.

Let us recall first the Slobodeckij spaces $W^{s,p}(\Circle)$. For $m \in \NN$ and $1 \le p \le \infty$, the $W^{m,p}$-norm of a measurable functions on $\Circle$ is defined by
\begin{equation*}
  \norm{u}_{W^{m,p}} := \sum_{j=0}^{m}\norm{u^{(j)}}_{L^{p}},
\end{equation*}
where $u^{(j)}$ denotes the derivative of order $j$. When $s > 0$ is not an integer and $p < \infty$, the $W^{s,p}$-norm is defined by
\begin{equation*}
  \norm{u}_{W^{s,p}} := \norm{u}_{W^{m,p}} + \mathfrak{p}_{\sigma,p}(u^{(m)})
\end{equation*}
where
\begin{equation*}
  m = [s] \quad \text{and} \quad s = m + \sigma,
\end{equation*}
and the semi-norm $\mathfrak{p}_{\sigma,p}$ ($0 < \sigma < 1$) is defined by
\begin{equation*}
    \mathfrak{p}_{\sigma,p}(w) = \left(\int_{\Circle}\int_{\Circle}\frac{\vert w(x) - w(y)\vert^{p}}{\vert x-y\vert^{1+p\sigma}}\,dx\,dy\right)^{1/p}.
\end{equation*}
In the following, we write $\mathfrak{p}_{\sigma,2} = \mathfrak{p}_{\sigma}$, when there is no ambiguity. The Banach space $W^{s,p}(\Circle)$ is by definition the completion of $\CS$ with respect to the $W^{s,p}$-norm and when $p = 2$, we get the Hilbert space
\begin{equation*}
  W^{s,2}(\Circle) = \HH{s}.
\end{equation*}
Recall that for $1 \le p, q < \infty$ and $r, s \in \RR$ such that
\begin{equation*}
  r \ge s, \quad \text{and} \quad r - \frac{1}{p} \ge s - \frac{1}{q},
\end{equation*}
we have the continuous \emph{Sobolev embeddings} (see \cite{Tri1983})
\begin{equation*}
  W^{r,p}(\Circle) \subset W^{s,q}(\Circle),
\end{equation*}
and
\begin{equation*}
  \HH{s} \subset W^{m,\infty}(\Circle),
\end{equation*}
for $s = m + \sigma$, $m \in \NN$ and $\sigma > 1/2$.

Recall that the space $\HH{s}$ is a multiplicative algebra for $s > 1/2$. We have, moreover, the following result which is a discrete version of~\cite[Lemma 2.3]{IKT2013}.

\begin{lem}\label{lem:pointwise_multiplication}
Let $q > 1/2$ and $0 \le \rho \le q$. Then, pointwise multiplication in $\CS$ extends to a continuous bilinear mapping
\begin{equation*}
  \HH{q} \times \HH{\rho} \to \HH{\rho}.
\end{equation*}
\end{lem}

\begin{proof}
Given $u,v\in \CS$, we have
\begin{equation*}
  \norm{uv}_{H^{\rho}}^{2} = \sum_{n \in \ZZ} (1+n^{2})^{\rho} \abs{\widehat{uv}(n)}^{2}.
\end{equation*}
Let us introduce the sequences
\begin{equation*}
  \tilde{u}(k) := (1+ k^{2})^{q/2}\abs{\hat{u}(k)}, \quad \text{and} \quad \tilde{v}(l) := (1+ l^{2})^{\rho/2}\abs{\hat{v}(l)},
\end{equation*}
so that
\begin{equation*}
  \norm{\tilde{u}}_{l^{2}} = \norm{u}_{H^{q}}, \quad \text{and} \quad \norm{\tilde{v}}_{l^{2}} = \norm{v}_{H^{\rho}}.
\end{equation*}
We have thus
\begin{align*}
  (1+n^{2})^{\rho/2} \abs{\widehat{uv}(n)} & = (1+n^{2})^{\rho/2} \abs{\sum_{k+l = n} \hat{u}(k)\hat{v}(l)}
  \\
  & \lesssim \sum_{\substack{k+l = n \\ \abs{k} \le \abs{l}}} \frac{1}{(1+k^{2})^{q/2}} \tilde{u}(k) \tilde{v}(l) + \sum_{\substack{k+l = n \\ \abs{k} > \abs{l}}} \frac{1}{(1+l^{2})^{q/2}} \tilde{u}(k) \tilde{v}(l)
  \\
  & \lesssim (\lambda^{q}\tilde{u} \ast \tilde{v})(n) + (\tilde{u} \ast \lambda^{q}\tilde{v})(n),
\end{align*}
where $\lambda^{q}$ is the sequence defined by
\begin{equation*}
  \lambda^{q}(k) := \frac{1}{(1+k^{2})^{q/2}}
\end{equation*}
and $\ast$ stand for the convolution of sequences. By virtue of the Young inequality
\begin{equation*}
  \norm{a \ast b}_{l^{2}} \lesssim \norm{a}_{l^{1}} \norm{b}_{l^{2}},
\end{equation*}
which is valid for any complex sequences $a$,$b$, we get
\begin{equation*}
  \norm{\lambda^{q}\tilde{u} \ast \tilde{v}}_{l^{2}}
  \lesssim \norm{\lambda^{q}\tilde{u}}_{l^{1}} \norm{\tilde{v}}_{l^{2}},
  \quad
  \norm{\tilde{u} \ast \lambda^{q}\tilde{v}}_{l^{2}}
  \lesssim \norm{\lambda^{q}\tilde{v}}_{l^{1}} \norm{\tilde{u}}_{l^{2}},
\end{equation*}
and by Cauchy-Schwarz, we obtain finally
\begin{equation*}
  \norm{uv}_{H^{\rho}} \lesssim \norm{\lambda^{q}}_{l^{2}} \norm{\tilde{u}}_{l^{2}}\norm{\tilde{v}}_{l^{2}}
  \lesssim \norm{u}_{H^{q}}\norm{v}_{H^{\rho}},
\end{equation*}
which achieves the proof.
\end{proof}

The main estimates which have been used in this paper concerning the right representation of $\D{q}$ on $\mathcal{L}(\HH{\rho},\HH{\rho})$ are given below. Note that the cases overlap.

\begin{lem}\label{lem:Rphi_estimate}
Let $q > 3/2$, $\varphi\in\D{q}$ and $0 \le \rho \le q$.
\begin{description}
  \item[Case 1] For $0 \le \rho \le 1$, we have
    \begin{equation}\label{eq:Rphi_estimate_firstcase}
        \norm{R_{\varphi}}_{\mathcal{L}(\HH{\rho},\HH{\rho})} \le C_{\rho}\left(\norm{1/\varphi_{x}}_{L^{\infty}},\norm{\varphi_{x}}_{L^{\infty}}\right).
    \end{equation}
  \item[Case 2] For $0 \le \rho \le 2$, we have
    \begin{equation}\label{eq:Rphi_estimate_secondcase}
        \norm{R_{\varphi}}_{\mathcal{L}(\HH{\rho},\HH{\rho})} \le C_{\rho}\left(\norm{1/\varphi_{x}}_{L^{\infty}},\norm{\varphi_{x}}_{H^{q-1}}\right).
    \end{equation}
  \item[Case 3] For $3/2 < \rho \le 3$, we have
    \begin{equation}\label{eq:Rphi_estimate_thirdcase}
        \norm{R_{\varphi}}_{\mathcal{L}(\HH{\rho},\HH{\rho})} \le C_{\rho}\left(\norm{1/\varphi_{x}}_{L^{\infty}},\norm{\varphi_{x}}_{L^{\infty}}\right) \norm{\varphi_{x}}_{H^{\rho-1}}.
    \end{equation}
  \item[Case 4] For $\rho > 5/2$, we have
    \begin{equation}\label{eq:Rphi_estimate_fourthcase}
        \norm{R_{\varphi}}_{\mathcal{L}(\HH{\rho},\HH{\rho})} \le C_{\rho}\left(\norm{1/\varphi_{x}}_{L^{\infty}},\norm{\varphi_{x}}_{H^{\rho-2}}\right) \norm{\varphi_{x}}_{H^{\rho-1}}.
    \end{equation}
\end{description}
In each case, $C_{\rho}$ is a positive, continuous function on $(\RR^{+})^{2}$, depending on $\rho$.
\end{lem}

\begin{proof}
Let $q > 3/2$ and fix $\varphi \in \D{q}$. Since the estimates involve only linear expressions of $u$ and $\CS$ is dense in $\HH{\rho}$, it is enough to establish them for $u \in \CS$. Note first that a change of variables leads to
\begin{equation*}
  \norm{u \circ \varphi}_{L^{2}} \le \norm{1/\varphi_{x}}_{L^{\infty}}^{1/2} \norm{u}_{L^{2}},
\end{equation*}
whereas
\begin{equation*}
   \mathfrak{p}_{\sigma}(u\circ \varphi) \le \norm{1/\varphi_{x}}_{L^{\infty}}\norm{\varphi_{x}}_{L^{\infty}}^{(1+2\sigma)/2} \mathfrak{p}_{\sigma}(u),
\end{equation*}
for $\sigma \in (0,1)$. Moreover, we have
\begin{equation*}
  \norm{(u\circ \varphi)^{(1)}}_{L^{2}} \lesssim \norm{\varphi_{x}}_{L^{\infty}} \norm{1/\varphi_{x}}^{1/2}_{L^{\infty}} \norm{u_{x}}_{L^{2}},
\end{equation*}
which proves the \emph{first case}. Note also that this proves also~\eqref{eq:Rphi_estimate_secondcase} for $0 \le \rho \le 1$, because $\norm{\varphi_{x}}_{L^{\infty}} \le \norm{\varphi_{x}}_{H^{q-1}}$. Suppose now that $\rho = 1 + \sigma$, where $0 < \sigma \le 1$. In this case, we have
\begin{equation*}
\begin{split}
  \norm{u\circ \varphi}_{H^{\rho}} & \lesssim \norm{u\circ \varphi}_{H^{1}} + \norm{(u_{x}\circ \varphi)\varphi_{x}}_{H^{\sigma}}
  \\
  & \lesssim C_{1}\left( \norm{1/\varphi_{x}}_{L^{\infty}}, \norm{\varphi_{x}}_{H^{q-1}}\right) \norm{u}_{H^{1}}
  \\
  & \quad + C_{\sigma}\left( \norm{1/\varphi_{x}}_{L^{\infty}}, \norm{\varphi_{x}}_{H^{q-1}}\right) \norm{u_{x}}_{H^{\sigma}} \norm{\varphi_{x}}_{H^{q-1}},
\end{split}
\end{equation*}
where we have used lemma~\ref{lem:pointwise_multiplication} and~\eqref{eq:Rphi_estimate_secondcase} for $0 \le \rho \le 1$. This completes the proof of the \emph{second case}. If $\rho=1+\sigma$ with $1/2 < \sigma \le 1$, we have using~\eqref{eq:Rphi_estimate_firstcase} and the fact that $\HH{\sigma}$ is a multiplicative algebra
\begin{equation*}
\begin{split}
  \norm{u\circ \varphi}_{H^{\rho}} & \lesssim \norm{u\circ \varphi}_{H^{1}} + \norm{u_{x}\circ \varphi}_{H^{\sigma}} \norm{\varphi_{x}}_{H^{\sigma}} \\
  & \lesssim C_{1}\left( \norm{1/\varphi_{x}}_{L^{\infty}}, \norm{\varphi_{x}}_{L^{\infty}}\right) \norm{u}_{H^{1}}
  \\
  & \quad + C_{\sigma}\left( \norm{1/\varphi_{x}}_{L^{\infty}}, \norm{\varphi_{x}}_{L^{\infty}}\right) \norm{u_{x}}_{H^{\sigma}} \norm{\varphi_{x}}_{H^{\sigma}}.
\end{split}
\end{equation*}
Now, noting that
\begin{equation*}
  1 \le \norm{1/\varphi_{x}}_{L^{\infty}} \norm{\varphi_{x}}_{L^{\infty}} \le \norm{1/\varphi_{x}}_{L^{\infty}} \norm{\varphi_{x}}_{H^{\sigma}},
\end{equation*}
we get
\begin{equation*}
  \norm{u\circ \varphi}_{H^{\rho}} \lesssim C_{\rho}\left( \norm{1/\varphi_{x}}_{L^{\infty}}, \norm{\varphi_{x}}_{L^{\infty}}\right) \norm{\varphi_{x}}_{H^{\rho-1}} \norm{u}_{H^{\rho}},
\end{equation*}
where $C_{\rho}$ is a positive, continuous function on $(\RR^{+})^{2}$. If $\rho=2+\sigma$ where $0 < \sigma \le 1$, we have by virtue of lemma~\ref{lem:pointwise_multiplication} and~\eqref{eq:Rphi_estimate_firstcase}
\begin{equation*}
\begin{split}
  \norm{(u \circ \varphi)^{(2)}}_{H^{\sigma}} & \lesssim  \norm{u_{xx} \circ \varphi}_{H^{\sigma}} \norm{(\varphi_{x})^{2}}_{H^{1}} + \norm{u_{x} \circ \varphi}_{H^{1}} \norm{\varphi_{xx}}_{H^{\sigma}} \\
    & \lesssim C_{\sigma}\left( \norm{1/\varphi_{x}}_{L^{\infty}}, \norm{\varphi_{x}}_{L^{\infty}}\right) \norm{u_{xx}}_{H^{\sigma}} \norm{\varphi_{x}}_{L^{\infty}}\norm{\varphi_{x}}_{H^{1}}
    \\
    & \quad + C_{1}\left( \norm{1/\varphi_{x}}_{L^{\infty}}, \norm{\varphi_{x}}_{L^{\infty}}\right) \norm{u_{x}}_{H^{1}} \norm{\varphi_{xx}}_{H^{\sigma}}
\end{split}
\end{equation*}
because $\norm{(\varphi_{x})^{2}}_{H^{1}} \lesssim \norm{\varphi_{x}}_{L^{\infty}}\norm{\varphi_{x}}_{H^{1}}$. Therefore
\begin{equation*}
\begin{split}
  \norm{u\circ \varphi}_{H^{\rho}} & \lesssim \norm{u\circ \varphi}_{H^{1}} + \norm{(u\circ \varphi)^{(2)}}_{H^{\sigma}} \\
    & \lesssim C_{\rho}\left( \norm{1/\varphi_{x}}_{L^{\infty}}, \norm{\varphi_{x}}_{L^{\infty}}\right) \norm{\varphi_{x}}_{H^{\rho-1}} \norm{u}_{H^{\rho}},
\end{split}
\end{equation*}
where $C_{\rho}$ is a positive, continuous function. This proves \emph{the third case}. If $\rho=2+\sigma$ and $1/2 < \sigma \le 1$ we get immediately \eqref{eq:Rphi_estimate_fourthcase} from the preceding computation because then $\norm{\varphi_{x}}_{L^{\infty}} \lesssim \norm{\varphi_{x}}_{H^{\sigma}}$. Suppose now that $\rho=m+\sigma$, where $m\ge 3$ and $\sigma\in [0,1)$. Given an integer $n$, we have
\begin{equation*}
  (u \circ \varphi)^{(n+1)} = \sum_{k=0}^{n} \left( u_{x}^{(k)}\circ \varphi \right) W_{n,k}(\varphi),
\end{equation*}
where $W_{n,k}(\varphi)$ is a homogeneous polynomial of degree $k+1$ in the variables $\varphi_{x}, \dotsc , \varphi_{x}^{(n-k)}$. The sequence $W_{n,k}$ is defined by
\begin{equation*}
  W_{n,n}(\varphi) = \varphi_{x}^{n+1}, \qquad W_{n,0}(\varphi) = \varphi_{x}^{(n)},
\end{equation*}
and
\begin{equation*}
  W_{n+1,k}(\varphi) = W_{n,k-1}(\varphi)\varphi_{x} + W_{n,k}(\varphi)^{\prime}, \quad 1 \le k \le n.
\end{equation*}
In particular, for $n \ge 2$, we have
\begin{equation*}
  W_{n,1}(\varphi) = (n+1) \varphi_{x}\varphi_{x}^{(n-1)} + P_{n}(\varphi_{x}, \dotsc , \varphi_{x}^{(n-2)}),
\end{equation*}
where $P_{n}$ is a homogeneous polynomial of degree 2. Thus, for $m \ge 2$, we get
\begin{equation}\label{eq:m+1_derivative}
\begin{split}
  \norm{(u\circ \varphi)^{(m+1)}}_{L^{2}} & \lesssim \norm{u_{x}}_{L^{\infty}}\norm{\varphi_{x}^{(m)}}_{L^{2}} + \norm{u_{xx}\circ\varphi}_{L^{2}}\norm{W_{m,1}(\varphi)}_{L^{\infty}}
  \\
  & \quad + \sum_{k=2}^{m} \norm{u_{x}^{(k)}\circ \varphi}_{L^{2}}\norm{W_{m,k}(\varphi)}_{L^{\infty}}
  \\
  & \lesssim \left[1 + \norm{1/\varphi_{x}}_{L^{\infty}}\left(\norm{\varphi_{x}}_{L^{\infty}} + \sum_{k=1}^{m} \norm{\varphi_{x}}_{H^{m-1}}^{k}\right)\right]
  \\
  & \quad \times \norm{\varphi_{x}}_{H^{m-1}}\norm{u}_{H^{m}},
\end{split}
\end{equation}
because
\begin{equation*}
  \norm{W_{m,1}(\varphi)}_{L^{\infty}} \lesssim \norm{\varphi_{x}}_{L^{\infty}} \norm{\varphi_{x}}_{H^{m}} + \norm{\varphi_{x}}^{2}_{H^{m-1}},
\end{equation*}
and
\begin{equation*}
  \norm{W_{m,k}(\varphi)}_{L^{\infty}} \lesssim \norm{\varphi_{x}}^{k+1}_{H^{m-1}}, \quad 2 \le k \le m.
\end{equation*}
Starting, with $m=2$, we get first \eqref{eq:Rphi_estimate_fourthcase} for $\rho = 3$ and the case when $\rho$ is an integer $m \ge 3$ is obtained by an inductive argument, using~\eqref{eq:m+1_derivative}. Now, using lemma~\ref{lem:pointwise_multiplication}, we have for $m \ge 3$ and $0 < \sigma <1$
\begin{equation*}
  \norm{(u \circ \varphi)^{(m)}}_{H^{\sigma}}
  \lesssim \norm{u_{x}\circ \varphi}_{H^{1}}\norm{\varphi_{x}^{(m-1)}}_{H^{\sigma}} + \sum_{k=1}^{m-1} \norm{u_{x}^{(k)}\circ \varphi}_{H^{\sigma}} \norm{W_{m-1,k}(\varphi)}_{H^{1}}.
\end{equation*}
But, for $2 \le k \le m-1$ we have
\begin{equation*}
  \norm{W_{m-1,k}(\varphi)}_{H^{1}} \lesssim \norm{\varphi_{x}}_{H^{m-2}}^{k+1}.
\end{equation*}
and for $k=1$ and $m-1 \ge 2$, we have
\begin{equation*}
  \norm{W_{m-1,1}(\varphi)}_{H^{1}} \lesssim \norm{\varphi_{x}}_{H^{m-2}} \norm{\varphi_{x}}_{H^{m-1}}.
\end{equation*}
We get therefore
\begin{equation*}
\begin{split}
  \norm{(u \circ \varphi)^{(m)}}_{H^{\sigma}}
  & \lesssim \left( \norm{1/\varphi_{x}}^{1/2}_{L^{\infty}} + \norm{1/\varphi_{x}}_{L^{\infty}}^{3/2} + \norm{1/\varphi_{x}}^{1/2}_{L^{\infty}} \sum_{k=1}^{m-1} \norm{\varphi_{x}}_{H^{m-2}}^{k} \right) \\
  & \quad \times \norm{\varphi_{x}}_{H^{\rho-1}} \norm{u}_{H^{\rho}}
\end{split}
\end{equation*}
which shows that estimate~\eqref{eq:Rphi_estimate_fourthcase} is also true when $\rho > 5/2$ is not necessary an integer. This proves the \emph{fourth case} and achieves the proof.
\end{proof}

\begin{cor}\label{cor:Rphi_continuity}
Let $q > 3/2$ and $0 \le \rho \le q$. Then, the mapping
\begin{equation*}
  (\varphi, u) \mapsto u \circ \varphi, \qquad \D{q} \times \HH{\rho} \to \HH{\rho}
\end{equation*}
is continuous.
\end{cor}

\begin{proof}
Lemma~\ref{lem:Rphi_estimate} shows that the mapping
\begin{equation*}
	\varphi \mapsto R_{\varphi}, \qquad \D{q} \to \mathcal{L}(\HH{\rho},\HH{\rho})
\end{equation*}
is locally bounded for $q > 3/2$ and $0 \le \rho \le q$. To establish the continuity of the mapping
\begin{equation*}
	(\varphi, u) \mapsto u \circ \varphi, \qquad \D{q} \times \HH{\rho} \to \HH{\rho},
\end{equation*}
it is thus sufficient to prove that
\begin{equation*}
	\varphi \mapsto u \circ \varphi, \qquad \D{q} \to \HH{\rho}.
\end{equation*}
is continuous for each fixed $u\in \HH{\rho}$. Let $\varphi_{0} \in \D{q}$ and $\eps >0$. Choose a neighbourhood $V$ of $\varphi_{0}$, that we may suppose to be a ball in a local chart of the Banach manifold $\D{q}$, and on which
\begin{equation*}
  \norm{R_{\varphi}}_{\mathcal{L}(\HH{\rho},\HH{\rho})} < K,
\end{equation*}
for some positive constant $K$. Since $\CS$ is dense in $\HH{\rho}$, we can find $w\in \CS$ such that
\begin{equation*}
  \norm{u-w}_{H^{\rho}} < \eps /K.
\end{equation*}
We have thus
\begin{equation*}
  \norm{u \circ \varphi - u \circ \varphi_{0}}_{H^{\rho}} < \norm{w \circ \varphi - w \circ \varphi_{0}}_{H^{\rho}} + 2 \eps.
\end{equation*}
Now, let $\varphi(t):= t\varphi + (1-t)\varphi_{0}$. We have, pointwise
\begin{equation*}
  w\circ\varphi - w\circ\varphi_{0} = \int_{0}^{1} (w_{x}\circ\varphi(t)) (\varphi - \varphi_{0})\,dt,
\end{equation*}
and we deduce from lemma~\ref{lem:pointwise_multiplication}, that
\begin{equation*}
  \norm{w\circ\varphi - w\circ\varphi_{0}}_{H^{\rho}} \le C \int_{0}^{1} \norm{w_{x}\circ\varphi(t)}_{H^{\rho}} \norm{\varphi - \varphi_{0}}_{H^{q}}\,dt,
\end{equation*}
for some positive constant $C$, which depends only on $q$ and $\rho$. Thus
\begin{equation*}
  \norm{w\circ\varphi - w\circ\varphi_{0}}_{H^{\rho}} \le CK \norm{w_{x}}_{H^{\rho}} \norm{\varphi - \varphi_{0}}_{H^{q}},
\end{equation*}
and therefore, if $\norm{\varphi - \varphi_{0}}_{H^{q}}$ is small enough, we get that
\begin{equation*}
  \norm{u \circ \varphi - u \circ \varphi_{0}}_{H^{\rho}} < 3 \eps,
\end{equation*}
which achieves the proof.
\end{proof}

\begin{rem}\label{rem:norm-continuity}
The mapping
\begin{equation*}
	(\varphi, u) \mapsto R_{\varphi}(u) := u \circ \varphi, \qquad \D{q} \times \HH{q} \to \HH{q}.
\end{equation*}
is continuous for $q > 3/2$ (see \cite{EM1970}), but that does not imply that the mapping
\begin{equation*}
  \varphi \mapsto R_{\varphi}, \qquad \D{q} \to \mathcal{L}(\HH{q})
\end{equation*}
is continuous with respect to the operator norm on $\mathcal{L}(\HH{q})$ (\emph{norm continuity}). \emph{Norm continuity} obviously implies \emph{continuity} but the converse is false. Indeed, a general result in the theory of  semigroups of linear operators states that a semigroup on a Banach space $E$ is norm continuous at $0$, if and only if, its infinitesimal generator is bounded on $E$, cf. \cite[Theorem 1.2]{Paz1983}. Let now $q>3/2$ and let $\tau_{s}$ be the rotation by the angle $s$ on $\Circle$. Then the representation of the group $\{R_{\tau_s}\,;\, s\in \mathbb{R}\}$ is continuous on $\HH{q}$. But it cannot be norm continuous, since its infinitesimal generator $D$ is not bounded on $\HH{q}$. A direct argument, which shows that $\norm{R_{\tau_{s}} - \mathrm{Id}}_{\mathcal{L}(\HH{q})}$ is bounded away from $0$ for all $s$ near $0$ is runs as follows: Let $s \in (-1/2,1/2)$ and $u_{s}$ be a periodic, bump function with support in $(k-s/2,k+s/2)$ ($k \in \ZZ$) with $\norm{u_{s}}_{L^{2}} = 1$. We have then
\begin{equation*}
  \norm{R_{\tau_{s}}u_{s} - u_{s}}_{\HH{q}}^{2}  = 2\norm{u_{s}}_{\HH{q}}^{2},
\end{equation*}
because $u_{s}$ and $R_{\tau_{s}}u$ are $\HH{q}$-orthogonal and $R_{\tau_{s}}$ is an $\HH{q}$-isometry. Hence
\begin{equation*}
  \norm{R_{\tau_{s}} - \mathrm{Id}}_{\mathcal{L}(\HH{q})} \ge \sqrt{2} \quad \text{for} \quad \frac{-1}{2} < s < \frac{1}{2},
\end{equation*}
which proves that the representation $\varphi \mapsto R_{\varphi}$ is not \emph{norm continuous}.
\end{rem}

Nevertheless, we have the following result.

\begin{cor}\label{cor:norm_continuity}
Let $q > 3/2$. Then, the mappings
\begin{equation*}
	\varphi \mapsto R_{\varphi}, \qquad \D{q} \to \mathcal{L}(\HH{q},\HH{q-1})
\end{equation*}
and
\begin{equation*}
	\varphi \mapsto R_{\varphi^{-1}}, \qquad \D{q} \to \mathcal{L}(\HH{q},\HH{q-1})
\end{equation*}
are continuous.
\end{cor}

\begin{proof}
Note first that since $\varphi \mapsto \varphi^{-1}$ from $\D{q}$ to $\D{q}$ is continuous (see for instance~\cite{IKT2013}), it is enough to show that
\begin{equation*}
	\varphi \mapsto R_{\varphi}, \qquad \D{q} \to \mathcal{L}(\HH{q},\HH{q-1})
\end{equation*}
is continuous. Let $\varphi_{0} \in \D{q}$. Choose a neighbourhood $V$ of $\varphi_{0}$, that we may suppose to be a ball in a local chart of the Banach manifold $\D{q}$, and on which
\begin{equation*}
  \norm{R_{\varphi}}_{\mathcal{L}(\HH{q-1},\HH{q-1})} < K,
\end{equation*}
for some positive constant $K$. Now, let $\varphi(t):= t\varphi + (1-t)\varphi_{0}$ and $v \in \HH{q}$. We have, pointwise
\begin{equation*}
  v\circ\varphi - v\circ\varphi_{0} = \int_{0}^{1} (v_{x}\circ\varphi(t)) (\varphi - \varphi_{0})\,dt.
\end{equation*}
Hence
\begin{equation*}
  \norm{v\circ\varphi - v\circ\varphi_{0}}_{H^{q-1}} \le C \int_{0}^{1} \norm{v_{x}\circ\varphi(t)}_{H^{q-1}} \norm{\varphi - \varphi_{0}}_{H^{q-1}}\,dt,
\end{equation*}
for some positive constant $C$ and from which we deduce that
\begin{equation*}
  \norm{R_{\varphi} - R_{\varphi_{0}}}_{\mathcal{L}(\HH{q},\HH{q-1})} \le CK \norm{\varphi - \varphi_{0}}_{H^{q}}.
\end{equation*}
This shows that
\begin{equation*}
	\varphi \mapsto R_{\varphi}, \qquad \D{q} \to \mathcal{L}(\HH{q},\HH{q-1})
\end{equation*}
is locally Lipschitz continuous and completes the proof.
\end{proof}

Another interpreting consequence of estimates~\ref{lem:Rphi_estimate} is the following results, which extends~\cite[Theorem 1.2]{IKT2013} beyond the critical exponent.

\begin{cor}\label{cor:composition_regularity}
Let $q > 3/2$. Then, the mappings
\begin{equation*}
  (\varphi, v) \mapsto v \circ \varphi, \qquad \D{q} \times \HH{q} \to \HH{q-1}
\end{equation*}
and
\begin{equation*}
  (\varphi, v) \mapsto v \circ \varphi^{-1}, \qquad \D{q} \times \HH{q} \to \HH{q-1}
\end{equation*}
are $C^{1}$.
\end{cor}

\begin{proof}
We are going to show that
\begin{equation*}
  (\varphi, v) \mapsto v \circ \varphi^{-1}, \qquad \D{q} \times \HH{q} \to \HH{q-1}
\end{equation*}
is $C^{1}$. The proof that the first mapping is $C^{1}$ is similar and easier. Observe first that if $\varphi(s)$ and $v(s)$ are smooth paths in $\DiffS$ and $\CS$ respectively with $\varphi(0)=\varphi$ and $v(0)=v$, we have
\begin{equation*}
  \partial_{s} \left[ v(s) \circ \varphi(s)^{-1} \right]_{s = 0} = \left(\delta v \circ\varphi^{-1}\right) - \left(\frac{v_{x}\circ\varphi^{-1}}{\varphi_{x}\circ\varphi^{-1}}\right) \left(\delta\varphi \circ \varphi^{-1}\right),
\end{equation*}
where $\delta\varphi = \partial_{s} \left.\varphi(s)\right\vert_{s = 0}$ and $\delta v = \partial_{s} \left. v(s)\right\vert_{s = 0}$. Let $U$ be a local chart in $\D{q}$, that we assume to be a ball in $\HH{q}$. Given $\varphi_{0},\varphi_{1}$ in $\DiffS \cap U$ and $v_{0},v_{1} \in \CS$, we set
\begin{equation*}
  \varphi(t):=(1-t)\varphi_{0}+t\varphi_{1}, \quad \text{and} \quad v(t):=(1-t)v_{0}+tv_{1},
\end{equation*}
for $t\in [0,1]$. We have therefore
\begin{multline}\label{eq:mean-value}
  v_{1} \circ \varphi_{1}^{-1} - v_{0} \circ \varphi_{0}^{-1} = \int_{0}^{1} \left((v_{1}-v_{0}) \circ \varphi(t)^{-1}\right)\, dt
   \\
   - \int_{0}^{1} \left(\frac{v_{x}(t)\circ\varphi(t)^{-1}}{\varphi_{x}(t)\circ\varphi(t)^{-1}}\right) \left((\varphi_{1}-\varphi_{0}) \circ \varphi(t)^{-1}\right)\, dt.
\end{multline}
Now, using lemma~\ref{lem:path_integral_continuity}, we observe that both sides of~\eqref{eq:mean-value} are continuous expressions of $\varphi_{k},v_{k}$ ($k=0,1$). Using a density argument, we conclude therefore, that the equality is still true in $\HH{q-1}$ if we take $\varphi_{0},\varphi_{1} \in \D{q}$ and $v_{0},v_{1} \in \D{q}$. Furthermore, the mapping
\begin{equation*}
  (\varphi,v) \mapsto \left[ (\delta \varphi, \delta v) \mapsto R_{\varphi^{-1}}\delta v - \left(\frac{v_{x}\circ\varphi^{-1}}{\varphi_{x}\circ\varphi^{-1}}\right)R_{\varphi^{-1}}\delta\varphi\right]
\end{equation*}
from
\begin{equation*}
  \D{q}\times \HH{q} \quad \text{to} \quad \mathcal{L}\left(\HH{q}\times\HH{q},\HH{q-1}\right)
\end{equation*}
is continuous by corollary~\ref{cor:norm_continuity} and the fact that $\HH{q-1}$ is a multiplicative algebra for $q > 3/2$. We conclude the proof using lemma~\ref{lem:mean_value_criteria}.
\end{proof}

To conclude this Appendix, we provide an estimate for the norm of $(\varphi^{-1})_{x}$ which might be useful, on its own.

\begin{lem}\label{lem:phi_inverse_estimate}
  Let $q > 3/2$. Then
  \begin{equation*}
    \norm{(\varphi^{-1})_{x}}_{H^{q-1}} \lesssim C_{q}(\norm{1/\varphi_{x}}_{L^{\infty}}, \norm{\varphi_{x}}_{H^{q-1}})
  \end{equation*}
  where $C_{q}$ is a positive, continuous function on $(\RR^{+})^{2}$.
\end{lem}

\begin{proof}
Given $\sigma \in (0,1)$, a change of variables leads to the estimate
\begin{equation*}
   \norm{w\circ \varphi^{-1}}_{H^{\sigma}} \lesssim \left(  \norm{\varphi_{x}}_{L^{\infty}}^{1/2} + \norm{\varphi_{x}}_{L^{\infty}} \norm{1/\varphi_{x}}_{L^{\infty}}^{(1+2\sigma)/2} \right) \norm{w}_{H^{\sigma}},
\end{equation*}
for any $w \in \HH{\sigma}$ and any $C^{1}$ diffeomorphism $\varphi$, whereas
\begin{equation*}
  \norm{w\circ \varphi^{-1}}_{H^{1}} \lesssim \left( \norm{\varphi_{x}}_{L^{\infty}}^{1/2} + \norm{\varphi_{x}}_{L^{\infty}}^{3/2} \right) \norm{1/\varphi_{x}}_{L^{\infty}} \norm{w}_{H^{1}},
\end{equation*}
for any $w \in \HH{1}$ and any $C^{1}$ diffeomorphism $\varphi$.

1) Suppose first that $q = 1 + \sigma$ and thus $\sigma > 1/2$. We get
\begin{equation*}
\begin{split}
  \norm{(\varphi^{-1})_{x}}_{H^{\sigma}} & = \norm{\frac{1}{\varphi_{x}}\circ \varphi^{-1}}_{H^{\sigma}} \\
    & \le \left( \norm{\varphi_{x}}_{L^{\infty}}^{1/2} + \norm{\varphi_{x}}_{L^{\infty}} \norm{1/\varphi_{x}}_{L^{\infty}}^{(1+2\sigma)/2} \right) \norm{1/\varphi_{x}}_{H^{\sigma}}.
\end{split}
\end{equation*}
But a direct computation shows that
\begin{equation*}
  \norm{1/\varphi_{x}}_{H^{\sigma}} \lesssim \norm{1/\varphi_{x}}_{L^{\infty}}^{2} \norm{\varphi_{x}}_{H^{\sigma}},
\end{equation*}
which finishes the proof of the lemma for $q < 2$, since $\norm{\varphi_{x}}_{L^{\infty}} \lesssim \norm{\varphi_{x}}_{H^{\sigma}}$.

2) Suppose now that $q = m + \sigma$, where $m \ge 2$ and $\sigma \in [0,1)$. Given an integer $n$, we have
\begin{equation*}
  \left(\frac{1}{\varphi_{x}\circ\varphi^{-1}}\right)^{(n)} = \frac{1}{(\varphi_{x}\circ\varphi^{-1})^{2n+1}} P_{n}(\varphi_{x}\circ\varphi^{-1}, \dotsc, \varphi_{x}^{(n)}\circ\varphi^{-1}),
\end{equation*}
where $P_{n}$ is a homogeneous polynomial of degree $n$, which is of partial degree at most one in $\varphi_{x}^{(n)}\circ\varphi^{-1}$. We have therefore
\begin{equation*}
\begin{split}
  \norm{\frac{1}{\varphi_{x}\circ\varphi^{-1}}}_{H^{m-2}} & \lesssim \sum_{k=0}^{m-2} \norm{1/\varphi_{x}}_{L^{\infty}}^{2k+1} \norm{P_{k}}_{L^{\infty}} \\
    & \lesssim \sum_{k=0}^{m-2} \norm{1/\varphi_{x}}_{L^{\infty}}^{2k+1} \norm{\varphi_{x}}_{H^{m-1}}^{k},
\end{split}
\end{equation*}
whereas
\begin{equation*}
  \norm{\left(\frac{1}{\varphi_{x}\circ\varphi^{-1}}\right)^{(m-1)}}_{H^{\sigma}} \lesssim \norm{1/\varphi_{x}}_{H^{1}}^{2m-1} \norm{P_{m-1}}_{H^{\sigma}},
\end{equation*}
by virtue of lemma~\ref{lem:pointwise_multiplication}. But
\begin{equation*}
  P_{m-1} = a_{0}(\varphi) + a_{1}(\varphi)\left(\varphi_{x}^{(m-1)}\circ\varphi^{-1}\right),
\end{equation*}
where $a_{j}(\varphi)$ is a homogeneous polynomial of degree $m-1-j$ in the variables $\varphi_{x}\circ\varphi^{-1}, \dotsc , \varphi_{x}^{(m-2)}\circ\varphi^{-1}$. Therefore, we have
\begin{equation*}
  \norm{P_{m-1}}_{H^{\sigma}} \lesssim \norm{a_{0}(\varphi)}_{H^{1}} + \norm{a_{1}(\varphi)}_{H^{1}} \norm{\varphi_{x}^{(m-1)}\circ\varphi^{-1}}_{H^{\sigma}},
\end{equation*}
where
\begin{equation*}
  \norm{a_{j}(\varphi)}_{H^{1}} \lesssim \left( \norm{\varphi_{x}}_{L^{\infty}}^{1/2} + \norm{\varphi_{x}}_{L^{\infty}}^{3/2} \right)^{m-1-j} \norm{1/\varphi_{x}}_{L^{\infty}}^{m-1-j} \norm{\varphi_{x}}_{H^{m-1}}^{m-1-j},
\end{equation*}
and
\begin{equation*}
  \norm{\varphi_{x}^{(m-1)}\circ\varphi^{-1}}_{H^{\sigma}} \lesssim \left(  \norm{\varphi_{x}}_{L^{\infty}}^{1/2} + \norm{\varphi_{x}}_{L^{\infty}} \norm{1/\varphi_{x}}_{L^{\infty}}^{(1+2\sigma)/2} \right) \norm{\varphi_{x}}_{H^{q-1}}
\end{equation*}
which ends the proof.
\end{proof}


\section*{Acknowledgments}

The authors wish to express their gratitude to the Erwin Schr\"{o}dinger International Institute for Mathematical Physics for providing an excellent research environment.

Finally it is a pleasure to thank Anders Melin and Elmar Schrohe for helpful discussions.


\end{document}